\documentclass[11pt,reqno,a4paper,]{article}
\usepackage{amssymb,amsmath}
\usepackage{amssymb,amsthm}
\usepackage{mathdots}
\usepackage{amsmath, amssymb} 
\usepackage{amscd}            
\usepackage{amsxtra}          
\usepackage{upref}   
\usepackage{enumerate}         
\usepackage{textcomp}
\usepackage[mathscr]{eucal}

\theoremstyle{plain}

\newtheorem{theorem}{Theorem}

\newtheorem{prop}[theorem]{Proposition}

\newtheorem{theo}[theorem]{Theorem}

\newtheorem{lemma}[theorem]{Lemma}

\newtheorem{definition}[theorem]{Definition}
\newtheorem{example}[theorem]{Example}

\theoremstyle{remark}

\numberwithin{theorem}{section}
\newcommand{\be}%
  {\protect\setcounter{equation}{\value{subsubsection}}}  
\newcommand{\ee}%

\newcommand{\F}{\mathbb{F}}

\usepackage{enumerate}

\numberwithin{Theorem}{section}
\numberwithin{Lemma}{section}

\numberwithin{Corollary}{section}
\numberwithin{Example}{section}
\numberwithin{Remark}{section}
\date{}
\allowdisplaybreaks{}
\begin{document}
\title{Cyclic codes over the ring $\F_p[u,v,w]/\langle u^2, v^2, w^2, uv-vu, vw-wv, uw-wu \rangle$}

\author {Pramod Kumar Kewat and Sarika Kushwaha 
\\\textit{Department of Applied Mathematics, Indian School of Mines,} \\\textit{Dhanbad-826 004, Jharkhand, India}\\}
\maketitle 
\footnotetext {\hspace{-.65cm}
\textit{E-mail addresses}: pramodkewat@gmail.com, sarika120189@gmail.com.\\
 }


\begin{abstract}
 In this paper, we  investigate  cyclic codes over the ring  $ \F_p[u,v,w]\langle u^2,$ $v^2, w^2$, $uv-vu, vw-wv, uw-wu \rangle$, where $p$ is a prime number. Which is a part of family of Frobenius rings. We find a unique set of generators for these codes and characterize the free cyclic codes. We also study the rank and the Hamming distance of these codes. We also constructs some good $p-ary$ codes as the Gray images of these cyclic codes.
 \end{abstract}
  \maketitle
\markboth{P.K. Kewat and S. Kushwaha}{Cyclic codes over the ring $R_{u^2,v^2,w^2,p}$}
\section{Introduction}

Cyclic codes are key families of linear codes because of their lavish algebraic structures and practical accomplishment. Codes over finite rings have been studied in the early 1970’s. A  considerable  attention has been given to codes over finite rings since 1990 because of their outstanding role in algebraic coding theory and their affluent applications. Recently, codes over some special finite rings mainly chain rings have been studied. Cyclic codes over different finite chain rings have been studied in \cite{Tah-Siap07},\cite{Ash-Ham11},\cite{Bonn-Udaya99},\cite{Dinh10},\cite{Dinh-Lopez04},\cite{Aks-Pkk13}. More latterly, cyclic codes over finite non-chain rings have also been contemplated. However, the analysis on non-chain rings seems to be challenging as the algebraic structure does not allow to give a nice and compact presentation of linear codes over these rings.  Yildiz and Karadeniz in \cite{Yil-Kar11} studied cyclic codes of odd length over non-chain ring $\F_2[u,v]/\langle u^2, v^2, uv-vu\rangle$ and found some good binary codes as the Gray images of these cyclic codes. Sobhani and Molakarimi in \cite{Sob-Mor13} extended these studies to cyclic codes over the ring $F_{2^m}[u,v]/\langle u^2,$ $v^2, uv-vu \rangle$. The authors of \cite{KGP15} extended these studies to cyclic codes over the ring $\F_p[u,v]/\langle u^2, v^2,$ $uv-vu \rangle$ and have found some good ternary codes as the Gray images of these cyclic codes. The authors of \cite{Dou-Yil-Kar12} have considered a family of rings $F_2[u_1,u_2,\cdots u_k]/\langle u_1^2,$ $u_2^2,\cdots, u_k^2, u_iu_j-u_ju_i \rangle$ and  characterized the nontrivial one-generator cyclic codes over these rings.\\
\indent
In this paper, we study the cyclic codes of arbitrary length $n$ over the ring $R_{u^2, v^2, w^2, p}$ = $\F_p[u,v,w]/\langle u^2, v^2, w^2, uv-vu, vw-wv, uw-wu \rangle$, where $p$ is a prime number. Note that the ring $R_{u^2, v^2, w^2, p}$ can also be viewed as the ring $ \F_p + u \F_p + v \F_p + uv \F_p+ w \F_p + uw \F_p + vw \F_p + uvw \F_p$, $ u^2=0$, $v^2=0$,$w^2=0$ and $uv = vu$, $vw=wv$, $uw=wu$. The techniques we have used to find a set of generators are similar to the techniques discussed in \cite{Tah-Siap07,KGP15,Aks-Pkk13}. Let C be a cyclic code over the ring $R_{u^2, v^2, w^2, p}$.  We view the cyclic code $C$ as an ideal in the ring $R_{u^2, v^2, w^2, p,n}$=$R_{u^2, v^2, w^2, p}/\left\langle x^n-1 \right\rangle$. Then we define the projection map from $R_{u^2, v^2, w^2, p,n}$ to $R_{u^2, v^2, p,n}$ and we get an ideal in the ring $R_{u^2, v^2, p,n}$. The structure of cyclic codes over the ring $R_{u^2, 
v^2, p}$ is known from \cite{KGP15}. By pullback, we find a set of generators for a cyclic code over the ring $R_{u^2, v^2,w^2 p}$. We also provide the characterization of the free cyclic codes over the ring $R_{u^2,v^2,w^2,p}$. When $n$ is relatively prime to $p$,  we get a simpler form for a set of generators of these cyclic codes. By using the division algorithm and direct computations, we find the rank and the minimal spanning set of these cyclic codes. We also find the Hamming distance of these codes for the length $p^l$. Again, the techniques we have used to find the minimum distance are similar to those discussed in  \cite{KGP15,Aks-Pkk13}.\\
\indent
This paper is organized as follows: In Section 2, we give some basic definitions and define a Gray map for a linear code over $R_{u^2, v^2, w^2, p}$. In Section 3, we find a unique set of generators along with the conditions on these generators. We also discuss here the generating polynomials for the case of a free cyclic code and $n$ relatively prime to $p$. In Section 4, we find rank and a minimal spanning set for these codes. In Section 5, we find the minimum distance of these codes for length $p^l$. In Section 6, we discuss some examples in which we construct some near optimal codes over $ \F_2, \F_3, \F_5 $ of length 32, 24, 40 respectively as the images of cyclic codes over the ring $R_{u^2, v^2, w^2, p}$ under the Gray map.

\section{Preliminaries}
A ring with the unique maximal ideal is called a local ring. Let $R$ be a finite commutative local ring with the maximal ideal $M$. Let $\overline{R} = R\textfractionsolidus M$ be the residue field and $\mu : R[x] \rightarrow \overline{R}[x]$ denote the natural ring homomorphism that maps $r \mapsto r + M$ and the variable $x$ to $x$. The degree of the polynomial $f(x) \in R[x]$ is defined as the degree of the polynomial $\mu(f(x))$ in $\overline{R}[x]$, i.e., $deg(f(x)) = deg(\mu(f(x))$ (see, for example, \cite{McDonald74}). A polynomial $f(x) \in R[x]$ is called regular if it is not a zero divisor.
The following conditions are equivalent for a finite commutative local ring $R$.
\begin{prop} {\rm (cf. \cite[Exercise XIII.2(c)]{McDonald74})} \label{regular-poly}
Let $R$ be a finite commutative local ring. Let $f(x) = a_0+a_1x+ \cdots +a_nx^n$ be in $R[x]$, then the following are equivalent.
\begin{enumerate}[{\rm (1)}]
\item $f(x)$ is regular; \label{1}
\item  $\langle a_0, a_1, \cdots , a_n \rangle = R$; \label{4} 
\item  $a_i$ is an unit for some $i$, $0 \leq i \leq n$; \label{3}
\item  $\mu(f(x)) \neq 0$; \label{2}
\end{enumerate}
\end{prop}
The following version of the division algorithm holds true for polynomials over finite commutative local rings.
\begin{prop}
let $R$ be a finite commutative local ring. Let f(x) and g(x) be non zero polynomials in $R[x]$. If $g(x)$ is regular, then there exist polynomials $q(x)$ and $r(x)$ in $R[x]$ such that $f(x)=g(x)q(x)+r(x)$ and $deg(r(x))<deg(g(x))$. 
\end{prop}
\subsection{The Ring $R_{u^2, v^2,w^2, p}$}

Let $R_{u^2, v^2, w^2, p}$ =$ \F_p + u \F_p + v \F_p + uv \F_p+ w \F_p + uw \F_p + vw \F_p + uvw \F_p$, $ u^2=0$, $v^2=0$,$w^2=0$ and $uv = vu$, $vw=wv$, $uw=wu$,
 It is easy to see that the ring $R_{u^2, v^2, w^2, p}$ is a finite local ring with the unique maximal ideal $\langle u,v,w\rangle$. Let $f_1(x)$ be a non zero polynomial in $\F_p[x]$. By Proposition \ref{regular-poly}, it is also easy to see that the polynomial $f_1(x)+uf_{1,2}(x)+vf_{1,3}(x)+uvf_{1,4}(x)+wf_{1,5}(x)+uwf_{1,6}(x)+vwf_{1,7}(x)+uvwf_{1,8}(x)  \in R_{u^2, v^2, w^2, p}[x]$ is regular. Note that $\text{deg}(f_1(x)+uf_{1,2}(x)+vf_{1,3}(x)+uvf_{1,4}(x)+wf_{1,5}(x)+uwf_{1,6}(x)+vwf_{1,7}(x)+uvwf_{1,8}(x))= \text{deg}(f_1(x))$.

\subsection{ The Gray map} 
Let $w_L$ and $w_H$ denote the Lee weight and the Hamming weight respectively. We define the Lee weight as follows:
\[w_L(\alpha)=w_H(\phi_L(\alpha)), ~\text{for all} ~ \alpha \in R_{u^2,v^2,w^2,p},\]
where the Gray map $\phi_L:R_{u^2, v^2, w^2, p}\rightarrow \F_p^{8}$ is defined as follows:\\
$\phi_L(\alpha_1+u\alpha_2+v\alpha_3+uv\alpha_4+w\alpha_5+uw\alpha_6+vw\alpha_7+uvw\alpha_8)=(\alpha_8,\alpha_6+\alpha_8,\alpha_7+\alpha_8,\alpha_4+\alpha_8,\alpha_5+\alpha_6+\alpha_7+\alpha_8,\alpha_2+\alpha_4+\alpha_6+\alpha_8,\alpha_3+\alpha_4+\alpha_7+\alpha_8,\alpha_1+\alpha_2+\alpha_3+\alpha_4+\alpha_5+\alpha_6+\alpha_7+\alpha_8)$\\
The Gray map naturally extend to $R_{u^2, v^2, w^2, p}^n$ as distance preserving isometry \\$\phi_L:(R_{u^2, v^2, w^2, p}^n, ~ \text{Lee weight}) \rightarrow (\F_p^{8n}, ~ \text{Hamming weight})$ as follows
\[\phi_L(\alpha_1, \alpha_2, \cdots, \alpha_n)\rightarrow (\phi_L(\alpha_1),\phi_L(\alpha_2), \cdots, \phi_L(\alpha_n)), ~\forall ~ a_i \in R_{u^2, v^2, w^2, p}.\]
By linearity of the map $\phi_L$ we obtain the following theorem.
\begin{theo}
If $C$ is a linear code over $R_{u^2, v^2, w^2, p}$ of length $n$, size $p^k$ and minimum lee weight $d$, then $\phi_L(C)$ is a $p$-ary linear code with parameters $[8n,k,d]$.
\end{theo}

\section{The structures of cyclic codes over the ring $R_{u^2, v^2,w^2, p}$}
 Let $R_{u^2,v^2,w^2,p} = \F_p[u, v, w]\textfractionsolidus \langle u^2,v^2,w^2,uv-vu,wv-vw,uw-wu \rangle$, where $p$ is a prime number and $n$ is a positive integer. We can write $R_{u^2,v^2,w^2,p}$ as $R_{u^2,v^2,w^2,p}= R_{u^2, v^2,p} + wR_{u^2, v^2,p}, w^2=0$, where $R_{u^2, v^2,p}=\F_p + u\F_p +v\F_p + uv\F_p$ and $u^2=0,v^2=0$. Let $R_{u^2,v^2,w^2,p,n}=R_{u^2,v^2,w^2,p}[x]\textfractionsolidus \langle x^n-1\rangle$. Let $C$ be a cyclic code of length $n$ over the ring $R_{u^2,v^2,w^2,p}$. We can also consider $C$ as an ideal in the ring $R_{u^2,v^2,w^2,p,n}$. We define the map $\psi : R_{u^2,v^2,w^2,p} \rightarrow R_{u^2,v^2,p}$ by $\psi(\alpha + w \beta) = \alpha$, where $\alpha, \beta \in R_{u^2,v^2,p}$. Clearly the map $\psi$ is a surjective ring homomorphism. Let $R_{u^2,v^2,p,n}=R_{u^2,v^2,p}[x]\textfractionsolidus \langle x^n-1\rangle$. We extend this homomorphism to a homomorphism $\phi$ from $C$ to the ring $R_{u^2,v^2,p,n}$ 
defined by
\begin{equation} \label{surj-hom}
\phi (c_0+c_1x+\cdots+c_{n-1}x^{n-1})=\psi(c_0)+\psi(c_1)x+\cdots+\psi(c_{n-1})x^{n-1},
\end{equation}
where $c_i \in R_{u^2,v^2,w^2,p}$. Let $J=\{r(x)\in R_{u^2,v^2,p,n}[x]: wr(x) \in \text{ker}\phi\}$. We see that $J$ is an ideal of $R_{u^2,v^2,p,n}$. Hence, we can consider $J$ as a cyclic code over $R_{u^2,v^2,p}$. We know from Theorem 3.1 of \cite{KGP15} that any ideal of $R_{u^2,v^2,p,n}$ is of the form $\langle g(x)+up_{1}(x)+vq_{1}(x)+uvr_{1}(x),ua_1(x)+vq_{2}(x)+uvr_{2},va_2(x)+uvr_{3}(x),uva_3(x)\rangle$. Now we assume that $B_1 =  g(x)+up_{1}(x)+vq_{1}(x)+uvr_{1}(x), B_2 = ua_1(x)+vq_{2}(x)+uvr_{2}, B_3 = va_2(x)+uvr_{3}(x)$, $B_{4} =uva_3(x)$. So $J = \langle B_1, B_2,B_3,B_4\rangle$. Therefore, we can write $\text{ker}\phi=\langle wB_1, wB_2,wB_3,wB_4\rangle$. Since $\phi$ is a surjective homomorphism, the image $\text{Im}\phi$ is an ideal of $R_{u^2,v^2,p,n}$. Hence, $\text{Im}\phi$ is a cyclic code over $R_{u^2,v^2,p}$. Again we can write $\text{Im}\phi$ as above. That is, $\text{Im}\phi=\langle B_1', B_2',B_3',B_4'\rangle$. Therefore, the code $C$ over the ring $R_{u^2,v^2,w^2,p}$ can be 
written as $C=\langle A_1, A_2, \cdots, A_{8}\rangle$, where, $A_i$'s are defined as follows:\\

\noindent
$A_1=f_1(x)+uf_{1,2}(x)+vf_{1,3}(x)+uvf_{1,4}(x)+wf_{1,5}(x)+uwf_{1,6}(x)+vwf_{1,7}(x)+uvwf_{1,8}(x)$, \\
$A_2=uf_2(x)+vf_{2,3}(x)+uvf_{2,4}(x)+wf_{2,5}(x)+uwf_{2,6}(x)+vwf_{2,7}(x)+uvwf_{2,8}(x)$,\\
$A_3=vf_3(x)+uvf_{3,4}(x)+wf_{3,5}(x)+uwf_{3,6}(x)+vwf_{3,7}(x)+uvwf_{3,8}(x)$, \\
$A_4=uvf_4(x)+wf_{4,5}(x)+uwf_{4,6}(x)+vwf_{4,7}(x)+uvwf_{4,8}(x)$, \\
$A_5=wf_5(x)+uwf_{5,6}(x)+vwf_{5,7}(x)+uvwf_{5,8}(x)$, \\
$A_6=uwf_6(x)+vwf_{6,7}(x)+uvwf_{6,8}(x)$, \\
$A_7=vwf_7(x)+uvwf_{7,8}(x)$, \\
$A_8=uvwf_8(x)$. \\
Throughout this paper we use $A_1, A_2, \cdots, A_{8}$ for above polynomials.\\

For an ideal $C$ of the ring $R_{u^2, v^2, w^2, p,n}=R_{u^2, v^2, w^2, p}[x]/\langle x^n-1\rangle$, we define the residue and the torsion of the ideal $C$ as
\begin{align*}
&\text{Res}(C)=\{a\in R_{u^2,v^2,p,n}|~\exists ~b\in R_{u^2,v^2,p,n}: a+wb\in C\} ~\text{and}\\
&\text{Tor}(C)=\{a\in R_{u^2,v^2,p,n}|~wa\in C\}
\end{align*}
It can be easily shown that when $C$ is an ideal of $R_{u^2,v^2,w^2,p,n}$, $\text{Res}(C)$ and $\text{Tor}(C)$ both are ideals of $R_{u^2,v^2,p,n}$. And also it is easy to show that $\text{Res}(C)=\text{Im}\phi$ and $\text{Tor}(C)=J$. Now we define eight ideals associated to C.
\begin{align}
&C_1=\text{Res}(\text{Res}(\text{Res}(C))) = C ~\text{mod} \langle u,v,w\rangle=\langle f_1(x)\rangle\\ 
&C_2=\text{Tor}(\text{Res}(\text{Res}(C))) = \{f(x)\in \ \F_p[x]~|~uf(x)\in C ~\text{mod}~\langle v,w\rangle\}=\langle f_2(x)\rangle\\
&C_3=\text{Res}(\text{Tor}(\text{Res}(C))) = \{f(x)\in \ \F_p[x]~|~vf(x)\in C ~\text{mod}~ \langle uv,w\rangle\}=\langle f_3(x)\rangle\\
&C_4=\text{Tor}(\text{Tor}(\text{Res}(C))) = \{f(x)\in \ \F_p[x]~|~uvf(x)\in C ~\text{mod} ~\langle w\rangle\}=\langle f_4(x)\rangle\\
&C_5=\text{Res}(\text{Res}(\text{Tor}(C))) = \{f(x)\in \ \F_p[x]~|~wf(x)\in C ~\text{mod} ~\langle uw,vw\rangle\}=\langle f_5(x)\rangle\\
&C_6=\text{Tor}(\text{Res}(\text{Tor}(C))) = \{f(x)\in \ \F_p[x]~|~uwf(x)\in C ~\text{mod} ~\langle vw\rangle\}=\langle f_6(x)\rangle\\
&C_7=\text{Res}(\text{Tor}(\text{Tor}(C))) = \{f(x)\in \ \F_p[x]~|~vwf(x)\in C ~\text{mod} ~ \langle uvw\rangle\}=\langle f_7(x)\rangle\\
&C_8=\text{Tor}(\text{Tor}(\text{Tor}(C))) = \{f(x)\in \ \F_p[x]~|~uvwf(x)\in C\}=\langle f_8(x)\rangle. \label{c8}
\end{align}
These are ideals of $\ \F_p[x]/\langle x^n-1\rangle$, hence principal ideals. Throughout this paper we use $C_1, ~C_2, \cdots, C_{8}$ for above ideals.
\begin{theorem} \label{unique}
Any ideal $C$ of the ring $R_{u^2, v^2, w^2, p,n}$ is uniquely generated by the polynomials
$A_1, A_2, \cdots, A_8$ with $f_{i,j}(x)=0$ or $deg(f_{i,j}(x)) < deg(f_j(x))$, where $A_i, f_i$ and $f_{i,j}$ are defined as above.
\end{theorem}
\begin{proof}
 All these generators are chosen in such a way that all satisfies one of the conditions, i.e., either $ f_{i,j}(x)=0 $ or $\text{deg}(f_{i,j}(x))< \text{deg}(f_j(x))$. We only prove these conditions for $i=1$ and $2 \leq j \leq 8$. All other can be shown in a similar way.
Let $A_1(x)\neq 0$ ~and~ $\text{deg}(f_{1,2}(x))$ $ \geq \text{deg}(f_2(x))$. Then by dividing $f_{1,2}(x)$ by ~$ f_2(x)$, we   have $f_{1,2}(x)=a_1(x)f_2(x)+r_1(x)$, where $\text{deg}(r_1(x))<\text{deg}(f_2(x))$ or $r_1(x)=0$.
Now $A_1(x)-a_1(x)A_2(x)=f_1(x)+ur_1(x)+v(f_{1,3}(x)-a_1(x)f_{2,3}(x))+uv(f_{1,4}(x)-a_1(x)f_{2,4}(x))+w(f_{1,5}(x)-a_1(x)f_{2,5}(x))+uw(f_{1,6}(x)-a_1(x)f_{2,6}(x))+vw(f_{1,7}(x)-a_1(x)f_{2,7}(x))+uvw$ $(f_{1,8}(x)-a_1(x)f_{2,8}(x))\in C$.
If $\text{deg}(f_{1,3}(x)-a_1(x)f_{2,3}(x))$ $ \geq$ $ \text{deg}(f_3(x))$. Then  by  division  algorithm,  we  have
     $f_{1,3}(x)-a_1(x)f_{2,3}(x)=a_2(x)f_3(x)+r_2(x)$ where $\text{deg}(r_2(x))<\text{deg}(f_3(x))$ or $r_2(x)=0$.
Now $A_1(x)-a_1(x)A_2(x)-a_2(x)A_3(x)=f_1(x)+ur_1(x)+vr_2(x)+uv(f_{1,4}(x)-a_1(x)f_{2,4}(x)-a_2(x)f_{3,4}(x))+w(f_{1,5}(x)-a_1(x)f_{2,5}(x)-a_2(x)f_{3,5}(x))+uw(f_{1,6}(x)-a_1(x)f_{2,6}(x)-a_1(x)f_{3,6}(x))+vw(f_{1,7}(x)-a_1(x)f_{2,7}(x)-a_2(x)f_{3,7}(x))+uvw(f_{1,8}(x)-a_1(x)f_{2,8}(x)-a_2(x)f_{3,8}(x))\in C$.
If $\text{deg}(f_{1,4}(x)-a_1(x)f_{2,4}(x)-a_2(x)f_{3,4}(x))$ $\geq \text{deg}(f_4(x))$. Then again by  division  algorithm,  we have
   $f_{1,4}(x)-a_1(x)f_{2,4}(x)-a_2(x)f_{3,4}(x)$ = $a_3(x)f_4(x)$+$r_3(x)$ where $\text{deg}(r_3(x))<\text{deg}(f_4(x))$ or $r_3(x)=0$.
Now $A_1(x)-a_1(x)A_2(x)-a_2(x)A_3(x)-a_3(x)A_4(x)=f_1(x)+ur_1(x)+vr_2(x)+uvr_3(x))+w(f_{1,5}(x)-a_1(x)f_{2,5}(x)-a_2(x)f_{3,5}(x)-a_3(x)f_{4,5}(x))+uw(f_{1,6}(x)-a_1(x)f_{2,6}(x)-a_2(x)f_{3,6}(x)-a_3(x)f_{4,6}(x))+vw(f_{1,7}(x)-a_1(x)f_{2,7}(x)-a_2(x)f_{3,7}(x)-a_3(x)f_{4,7}(x))+uvw(f_{1,8}(x)-a_1(x)f_{2,8}(x)-a_2(x)f_{3,8}(x)-a_3(x)f_{4,8}(x))\in C$. If we continue this process, we get the polynomial $A_1(x)-a_1(x)A_2(x)-a_2(x)A_3(x)-a_3(x)A_4(x)-a_4(x)A_5(x)-a_5(x)A_6(x)-a_6(x)A_7(x)-a_7(x)A_8(x)=f_1(x)+ur_1(x)+vr_2(x)+uvr_3(x)+wr_4(x)+uwr_5(x)+vwr_6(x)+uvwr_7(x)\in C$, which satisfies the required properties of the theorem and also the polynomial $A_1$ can be replaced by this polynomial. Now we have to prove that the polynomials $A_i$'s are unique. Here again, we prove the uniqueness only for the polynomial $A_1$. Others are similar. If possible, let $A_1=f_1(x)+uf_{1,2}(x)+vf_{1,3}(x)+uvf_{1,4}(x)+wf_{1,5}(x)+uwf_{1,6}(x)+vwf_{1,7}(x)+uvwf_{1,8}(x)$ and $B_1=f_1(x)+uf_{1,2}'(x)+vf_{1,3}'(x)+uvf_{1,4}'(x)
+wf_{1,5}'(x)+uwf_{1,6}'(x)+vwf_{1,7}'(x)+uvwf_{1,8}'(x)$ be two polynomials with the same properties in $C$. Hence, $A_1-B_1=u(f_{1,2}(x)-f'_{1,2}(x))+v(f_{1,3}(x)-f'_{1,3}(x))+uv(f_{1,4}(x)-f'_{1,4}(x))+w(f_{1,5}(x)-f'_{1,5}(x))+uw(f_{1,6}(x)-f'_{1,6}(x))+vw(f_{1,7}(x)-f'_{1,7}(x))+uvw(f_{1,8}(x)-f'_{1,8}(x))$. We have $A_1-B_1\in C$ which implies that $f_{1,2}(x)-f'_{1,2}(x)\in C_2=\langle f_2(x)\rangle$. Previously, we have proved that the degrees of both $f_{1,2}(x)$ and $f'_{1,2}(x)$ are less than degree of $f_2(x)$. Hence, $\text{deg}(f_{1,2}(x)-f'_{1,2}(x))< \text{deg}(f_2(x))$. But $f_2(x)$ is the minimum degree polynomial in $C_2$, which implies that $f_{1,2}(x)-f'_{1,2}(x)=0$. This gives, $f_{1,2}(x)=f'_{1,2}(x)$. Now $A_1-B_1=v(f_{1,3}(x)-f'_{1,3}(x))+uv(f_{1,4}(x)-f'_{1,4}(x))+w(f_{1,5}(x)-f'_{1,5}(x))+uw(f_{1,6}(x)-f'_{1,6}(x))+vw(f_{1,7}(x)-f'_{1,7}(x))+uvw(f_{1,8}(x)-f'_{1,8}(x)))$. We have $A_1-B_1\in C$ which implies that $f_{1,3}(x)-f'_{1,3}(x)\in C_3=\langle g_3(x)\rangle$. Again, we have already 
proved that the degrees of $f_{1,3}(x)$ and $f'_{1,3}(x)$ are less than degree of $f_3(x)$. Hence, $\text{deg}(f_{1,3}(x)-f'_{1,3}(x))< \text{deg}(f_3(x))$, which implies that $f_{1,3}(x)-f'_{1,3}(x)=0$. This gives $f_{1,3}(x)=f'_{1,3}(x)$. Similarly, we can show that $f_{1,i}(x)=f'_{1,i}(x)$ for all $4\leq i\leq 8$. Thus, $A_1-B_1=0$. Hence, $A_1$ is unique.
\end{proof}
\begin{theorem} \label{properties}
 Let $C=\langle A_1, A_2, \cdots, A_8\rangle$ be an ideal of the ring $R_{u^2, v^2, w^2, p,n}$. Then we must have
\begin{enumerate}[{\rm (1)}]
\item $f_8(x)|f_i(x)$,  {\rm for} $1 \leq i\leq 7$;  $~~~f_j(x)|f_1(x)|(x^n-1)$,  {\rm for} $2 \leq j\leq 7$; 
\item $f_4(x)|f_2(x)$;$~~~f_4(x)|f_3(x)$;$~~~f_6(x)|f_5(x)$;
$~~~f_6(x)|f_2(x)$;$~~~f_7(x)|f_5(x)$;$~~~f_7(x)|f_3(x)~$;

\item$f_{i+1}(x)|f_{i,i+1}(x)\left(\frac{x^n-1}{f_{i}(x)}\right)$, {\rm for} $1 \leq i\leq 7$; 
\item${\rm{For ~ a ~ fix}} ~ j,  1 \leq j \leq 7, f_{i+j}(x)|\frac{x^n-1}{f_i(x)}\frac{x^n-1}{f_{i+1}(x)} \cdots \frac{x^n-1}{f_{i+j-1}(x)}f_{i,i+j}(x), ~ {\rm{for}} ~ 1 \leq i \leq 8-j$;
\item$f_i(x)|\frac{x^n-1}{f_{i-2}(x)}\left(f_{i-2,i}(x)-\frac{f_{i-2,i-1}(x)}{f_i-1(x)}f_{i-1,i}(x)\right)$, {\rm for} $3 \leq i\leq 8$; 
\item$f_i(x)|\frac{x^n-1}{f_{i-3}(x)}\left(f_{i-3,i}(x)-\frac{f_{i-3,i-2}(x)}{f_{i-2}(x)}f_{i-2,i}(x)-\left(\frac{f_{i-3,i-1}(x)-\frac{f_{i-3,i-2}(x)}{f_i-2(x)}f_{i-2,i-1}(x)}{f_{i-1}(x)}f_{i-1,i}(x)\right)\right)$;\\{\rm for} $4 \leq i\leq 8$;
\item$f_i(x)|\frac{x^n-1}{f_{i-4}(x)}\left(f_{i-4,i}(x)-\frac{f_{i-4,i-3}(x)}{f_{i-3}(x)}f_{i-3,i}(x)-Af_{i-2,i}(x)-Bf_{i-1,i}(x)\right)$; {\rm for} $i\in{(5,6,7,8)};$ $~~~~~~~~~~~$\\
{\rm where} $~~A=\left(\frac{f_{i-4,i-2}(x)-\frac{f_{i-4,i-3}(x)}{f_{i-3}(x)}f_{i-3,i-2}(x)}{f_{i-2}(x)}\right)$ {\rm and}  $B=\left(\frac{f_{i-4,i-1}(x)-\frac{f_{i-4,i-3}(x)}{f_{i-3}(x)}f_{i-3,i-1}(x)-Af_{i-2,i-1}(x)}{f_{i-1}(x)}\right)$. 
\item$f_i(x)|\frac{x^n-1}{f_{i-5}(x)}\left(f_{i-5,i}(x)-\frac{f_{i-5,i-4}(x)}{f_{i-4}(x)}f_{i-4,i}(x)-Af_{i-3,i}(x)-Bf_{i-2,i}(x)-Df_{i-1,i}(x)\right); $\\
 {\rm for} $i\in{(6,7,8)}$; 
{\rm where} $A=\left(\frac{f_{i-5,i-3}(x)-\frac{f_{i-5,i-4}(x)}{f_{i-4}(x)}f_{i-4,i-3}(x)}{f_{i-3}(x)}\right)$,\\
$B=\left(\frac{f_{i-5,i-2}(x)-\frac{f_{i-5,i-4}(x)}{f_{i-4}(x)}f_{i-4,i-2}(x)-Af_{i-3,i-2}(x)}{f_{i-2}(x)}\right)$ {\rm and}  \\
$D=\left(\frac{f_{i-5,i-1}(x)-\frac{f_{i-5,i-4}(x)}{f_{i-4}(x)}f_{i-4,i-1}(x)-Af_{i-3,i-1}(x)-Bf_{i-2,i-1}(x)}{f_{i-1}}\right)$.
\item$f_i(x)|\frac{x^n-1}{f_{i-6}(x)}\left(f_{i-6,i}(x)-Af_{i-5,i}(x)-Bf_{i-4,i}(x)-Df_{i-3,i}(x)-Ef_{i-2,i}(x)-Ff_{i-1,i}(x)\right)$ {\rm for} $i\in{(7,8)}$;\\
{\rm where} $A=\left(\frac{f_{i-6,i-5}(x)}{f_{i-5}(x)}\right)$, $B=\left(\frac{ f_{i-6,i-4}(x)-Af_{i-5,i-4}(x)}{f_{i-4}(x)}\right)$, $ D=\left(\frac{f_{i-6,i-3}(x)-Af_{i-5,i-3}(x)-Bf_{i-4,i-3}(x)}{f_{i-3}(x)}\right)$,\\
$E=\left(\frac{f_{i-6,i-2}(x)-Af_{i-5,i-2}(x)-Bf_{i-4,i-2}(x)-Df_{i-3,i-2}(x)}{f_{i-2}(x)}\right)$ {\rm and}  \\
$F=\left(\frac{f_{i-6,i-1}(x)-Af_{i-5,i-1}(x)-Bf_{i-4,i-1}(x)-Df_{i-3,i-1}(x)-Ef_{i-2,i-1}(x)}{f_{i-1}(x)}\right)$.
\item$f_8(x)|\frac{x^n-1}{f_{1}(x)}\left(f_{1,8}(x)-Af_{2,8}(x)-Bf_{3,8}(x)-Df_{4,8}(x)-Ef_{5,8}(x)-Ff_{6,8}(x)-Gf_{7,8}(x)\right),$
{\rm where}\\
$A=\left(\frac{f_{1,2}(x)}{f_2(x)}\right) ,$ $~~B=\left(f_{1,3}(x)-A\frac{f_{2,3}(x)}{f_{3}(x)}\right)$,$ D=\left(\frac{f_{1,4}(x)-Af_{2,4}(x)-Bf_{3,4}(x)}{f_{4}(x)}\right),$\\
$E=\left(\frac{f_{1,5}(x)-Af_{2,5}(x)-Bf_{3,5}(x)-Df_{4,5}(x)}{f_{5}(x)}\right)$, $F=\left(\frac{f_{1,6}(x)-Af_{2,6}(x)-Bf_{3,6}(x)-Df_{4,6}(x)-Ef_{5,6}(x)}{f_{6}(x)}\right)$ {\rm and}  
$G=\left(\frac{f_{1,7}(x)-Af_{2,7}(x)-Bf_{3,7}(x)-Df_{4,7}(x)-Ef_{5,7}(x)-Ff_{6,7}(x)}{f_{7}(x)}\right).$

\item$f_i(x)|f_{i-2,i-1}(x)$ {\rm for} $i \in{(4,6,8)}$;
\item $f_i(x)|\left(f_{1,2}(x)-\frac{f_1(x)}{f_{i-1}(x)}f_{i-1,i}(x)\right)$, {\rm for} $ i \in{(4,6,8)}$; 
\item$f_i(x)|\left(f_{i-5,i-4}(x)-\frac{f_{i-5}(x)}{f_{i-1}(x)}f_{i-1,i}(x)\right)$, {\rm for} $i \in{(7,8)}$; 
 
\item$ f_i(x)|\left( f_{i-6,i-4}(x)-\frac{f_{i-6}(x)}{f_{i-2}(x)}f_{i-2,i}(x)+\frac{\left(f_{i-6,i-5}(x)-\frac{f_{i-6}(x)}{f_{i-2}(x)}f_{i-2,i-1}(x)\right)}{f_{i-1}(x)}f_{i-1,i}(x)\right),$ {\rm for} $i\in{(7,8)}$;
\item$f_7(x)|f_{4,5}(x)$ {\rm and}  $f_7(x)|f_{3,5}(x);$
\item$f_8(x)|f_{2,5}(x);$
\item$f_8(x)|\left(f_{3,6}(x)-\frac{f_{3,5}(x)}{f_7(x)}f_{7,8}(x)\right);$
\item$f_8(x)|\left(f_{4,6}(x)-\frac{f_{4,5}(x)}{f_7(x)}f_{7,8}(x)\right);$
\item$f_8(x)|\left(f_{5,6}(x)-\frac{f_5(x)}{f_7(x)}f_{7,8}(x)\right);$
\item$f_8(x)|\left(f_{1,4}(x)-\frac{f_{1}(x)}{f_5(x)}f_{5,8}(x)-Af_{6,8}(x)-Bf_{7,8}(x)\right),$\\
{\rm where} $A=\left(\frac{f_{1,2}(x)-\frac{f_{1}(x)}{f_5(x)}f_{5,6}(x)}{f_{6}(x)}\right)$ {\rm and}  $B=\left(\frac{f_{1,3}(x)-\frac{f_{1}(x)}{f_5(x)}f_{5,7}(x)-(A)f_{6,7}(x)}{f_{6}(x)}\right).$
\end{enumerate}
\end{theorem}
\begin{proof}
\begin{enumerate}
\item We have $vwA_2\in C$. Therefore, $uvwf_2(x) \in C$. This gives, $f_2(x)\in C_8=\langle f_8(x)\rangle$. Thus, $f_8(x)|f_2(x)$. Similarly, if we take $uwA_3$, $wA_4$, $uvA_5$, $vA_6$ and $uA_7$ we get $f_8(x)|f_i(x)$,  for $3 \leq i\leq 7$.
\item  We have $vA_2\in C$. Therefore, $uvf_2(x) \in C~mod~w$. This gives, $f_2(x)\in C_4=\langle f_4(x)\rangle$. Thus, $f_4(x)|f_2(x)$. Similarly, if we take $uA_3$, $uA_5$,$wA_2$, $vA_5$, $wA_3$ and take mod by $ w,vw,vw,uvw,uvw$ respectively, we  get the other conditions of (2). 
\item For $1\leq i \leq 7$, we have $\frac{x^n-1}{f_i(x)}A_i\in C$. Therefore, $\frac{x^n-1}{f_i(x)}f_{i,i+1}(x)\in C_{i+1}=\langle f_{i+1}(x)\rangle$. Hence, $f_{i+1}(x)|\frac{x^n-1}{f_i(x)}f_{i,i+1}(x)$.
\item For $j=1$, Condition 4 is reduced to Condition 3. For $j=2$ and for $1 \leq i \leq 6$, we have $\frac{x^n-1}{f_i(x)}\frac{x^n-1}{f_{i+1}(x)}A_i \in C.$ This with Condition 3 gives, $\frac{x^n-1}{f_i(x)}\frac{x^n-1}{f_{i+1}(x)}f_{i,i+2} \in C_{i+2}=\langle f_{i+2}(x)\rangle$. Hence, $f_{i+2}(x)|\frac{x^n-1}{f_i(x)}\frac{x^n-1}{f_{i+1}(x)}f_{i,i+2}$. This proves the condition for $j=2$. Similarly for others value of $j$ we can prove the Condition 4.
\item For $i = 3$, we have $\left( \frac{x^n-1}{f_1(x)}A_1-\frac{x^n-1}{f_1(x)}\frac{f_{1,2}(x)}{f_2(x)}A_2\right)=v\frac{x^n-1}{f_1(x)}(f_{1,3}(x)-\frac{f_{1,2}(x)}{f_2(x)}f_{2,3}(x))+uv\frac{x^n-1}{f_1(x)}(f_{1,4}(x)-\frac{f_{1,2}(x)}{f_2(x)}f_{2,4}(x))+w\frac{x^n-1}{f_1(x)}(f_{1,5}(x)-\frac{f_{1,2}(x)}{f_2(x)}f_{2,5}(x))+uw\frac{x^n-1}{f_1(x)}(f_{1,6}(x)-\frac{f_{1,2}(x)}{f_2(x)}f_{2,6}(x))+vw\frac{x^n-1}{f_1(x)}(f_{1,7}(x)-\frac{f_{1,2}(x)}{f_2(x)}f_{2,7}(x))+uvw\frac{x^n-1}{f_1(x)}(f_{1,8}(x)-\frac{f_{1,2}(x)}{f_2(x)}f_{2,8}(x))\in C $.
 Since $v\frac{x^n-1}{f_1(x)}(f_{1,3}(x)-\frac{f_{1,2}(x)}{f_2(x)}f_{2,3}(x))\in C mod <uv,w>$. Therefore  $\frac{x^n-1}{f_1(x)}(f_{1,3}(x)-\frac{f_{1,2}(x)}{f_2(x)}f_{2,3}(x))\in C_3$ $\Rightarrow $ $f_3(x)|\left( \frac{x^n-1}{f_1(x)}(f_{1,3}(x)-\frac{f_{1,2}(x)}{f_2(x)}f_{2,3}(x))\right)$.
 Similarly we get the results for rest of the values of $i$.
\item For $i=4$, we have\\ $\left(\frac{x^n-1}{f_1(x)}A_1-\frac{x^n-1}{f_1(x)}\frac{f_{1,2}(x)}{f_2(x)}A_2+\frac{x^n-1}{f_1(x)f_3(x)}\left(f_{1,3}(x)-\frac{f_{1,2}(x)}{f_2(x)}f_{2,3}(x)\right)A_3\right)\in C$. \\Since, $uv\frac{x^n-1}{f_1(x)}\left(f_{1,4}(x)-\frac{f_{1,2}(x)}{f_2(x)}f_{2,4}(x)+\frac{\left(f_{1,3}(x)-\frac{f_{1,2}(x)}{f_2(x)}f_{2,3}(x)\right)}{f_3(x)}f_{3,4}(x)\right)\in C~mod~w$.\\Therefore, $\frac{x^n-1}{f_1(x)}\left(f_{1,4}(x)-\frac{f_{1,2}(x)}{f_2(x)}f_{2,4}(x)+\frac{\left(f_{1,3}(x)-\frac{f_{1,2}(x)}{f_2(x)}f_{2,3}(x)\right)}{f_3(x)}f_{3,4}(x)\right)\in C_4$.\\$ \Rightarrow$ $f_4(x)|\frac{x^n-1}{f_1(x)}\left(f_{1,4}(x)-\frac{f_{1,2}(x)}{f_2(x)}f_{2,4}(x)+\frac{\left(f_{1,3}(x)-\frac{f_{1,2}(x)}{f_2(x)}f_{2,3}(x)\right)}{f_3(x)}f_{3,4}(x)\right).$\\Similarly we get the results for rest of the values of $i$.
 \item  For $i=4$, we have  
$ \frac{x^n-1}{f_1(x)}\left(A_1-f_{1,2}(x)\frac{A_2}{f_2(x)}+\left(f_{1,3}(x)-\frac{f_{1,2}(x)}{f_2(x)}f_{2,3}(x)\right)\frac{A_3}{f_3(x)}\right)-\\\left(\frac{x^n-1}{f_1(x)} \left(f_{1,4}(x)-f_{1,2}(x)\frac{A_2}{f_2(x)}+\left(f_{1,3}(x)-\frac{f_{1,2}(x)}{f_2(x)}f_{2,3}(x)\right)\frac{A_3}{f_3(x)}\right)\frac{A_4}{f_4(x)}\right)\in C$.\\ Since, 
$w\frac{x^n-1}{f_1(x)}\left(f_{1,5}(x)-f_{1,2}(x)\frac{f_{2,5}(x)}{f_2(x)}+\left(f_{1,3}(x)-\frac{f_{1,2}(x)}{f_2(x)}f_{2,3}(x)\right)\frac{f_{3,5}(x)}{f_3(x)}\right)-\\\frac{x^n-1}{f_1(x)} \left(f_{1,4}(x)-f_{1,2}(x)\frac{f_{2,4}(x)}{f_2(x)}+\left(f_{1,3}(x)-\frac{f_{1,2}(x)}{f_2(x)}f_{2,3}(x)\right)\frac{f_{3,4}(x)}{f_3(x)}\right)\frac{f_{4,5}(x)}{f_4(x)}\in C mod(uw,vw).
$\\Therefore,
$ \frac{x^n-1}{f_1(x)}\left(f_{1,5}(x)-f_{1,2}(x)\frac{f_{2,5}(x)}{f_2(x)}+\left(f_{1,3}(x)-\frac{f_{1,2}(x)}{f_2(x)}f_{2,3}(x)\right)\frac{f_{3,5}(x)}{f_3(x)}\right)-\\\frac{x^n-1}{f_1(x)} \left(f_{1,4}(x)-f_{1,2}(x)\frac{f_{2,4}(x)}{f_2(x)}+\left(f_{1,3}(x)-\frac{f_{1,2}(x)}{f_2(x)}f_{2,3}(x)\right)\frac{f_{3,4}(x)}{f_3(x)}\right)\frac{f_{4,5}(x)}{f_4(x)}\in C_5.$\\Similarly we get the results for rest of the values of $i$.
\item   For  $i=6$, we have 
$~~\left(\frac{x^n-1}{f_{1}(x)}A_1-AA_2-BA_3+DA_4+EA_5\right)\in C$\\Since, 
$uw\left(\frac{x^n-1}{f_{1}(x)}f_{1,6}(x)-Af_{2,6}(x)-Bf_{3,6}(x)+Df_{4,6}(x)+Ef_{5,6}(x)\right)\in C~mod~vw$.\\
Therefore, $\frac{x^n-1}{f_{1}(x)}\left(f_{1,6}(x)-\frac{f_{1,2}(x)}{f_{2}(x)}f_{2,6}(x)-Af_{3,6}(x)-Bf_{4,6}(x)-Df_{5,6}(x)\\\right)\in C_6 $\\$\Rightarrow$ $f_6(x)|\frac{x^n-1}{f_{1}(x)}\left(f_{1,6}(x)-\frac{f_{1,2}(x)}{f_{2}(x)}f_{2,6}(x)-Af_{3,6}(x)-Bf_{4,6}(x)-Df_{5,6}(x)\right). $
where, \\$A=\frac{x^n-1}{f_{1}(x)}\frac{f_{1,2}(x)}{f_{2}(x)}, B=\frac{x^n-1}{f_{1}(x)}\left(\frac{f_{1,3}(x)-\frac{f_{1,2}(x)}{f_{2}(x)}f_{2,3}(x)}{f_{3}(x)}\right)$,
$C=\frac{x^n-1}{f_{1}(x)}\left(\frac{-f_{1,4}(x)+\frac{f_{1,2}(x)}{f_{2}(x)}f_{2,4}(x)+Af_{3,4}(x)}{f_{4}(x)}\right)$\\and
$D=\frac{x^n-1}{f_{1}(x)}\left(\frac{-f_{1,5}(x)+\frac{f_{1,2}(x)}{f_{2}(x)}f_{2,5}(x)+Af_{3,5}(x)+Bf_{4,5}(x)}{f_{5}}\right)$.\\Similarly we get the results for rest of the values of $i$.
\item For  $i=7$, we have   
$~\frac{x^n-1}{f_{1}(x)}\left(A_1-A A_2-B A_3-D A_4-E A_5-F A_6\right)\in C$. Since, \\
$vw\frac{x^n-1}{f_{1}(x)}\left(f_{1,7}(x)-A f_{2,7}(x)-B f_{3,7}(x)-D f_{4,7}(x)-E f_{5,7}(x)-F f_{6,7}(x)\right)\in Cmod~uvw$.\\Therefore,
 $\frac{x^n-1}{f_{1}(x)}\left(f_{1,7}(x)-A f_{2,7}(x)-B f_{3,7}(x)-D f_{4,7}(x)-E f_{5,7}(x)-F f_{6,7}(x)\right)\in C_7$ 
$\Rightarrow$  $f_7(x)|\frac{x^n-1}{f_{1}(x)}\left(f_{1,7}(x)-A f_{2,7}(x)-B f_{3,7}(x)-D f_{4,7}(x)-E f_{5,7}(x)-F f_{6,7}(x)\right)$.\\
where, $A=\left(\frac{f_{1,2}(x)}{f_{2}(x)}\right) ,B=\left(f_{1,3}(x)-A \frac{f_{2,3}(x)}{f_{3}(x)}\right), D=\left(\frac{f_{1,4}(x)-A f_{2,4}(x)-B f_{3,4}(x)}{f_{4}(x)}\right)$ ,\\
$E=\left(\frac{f_{1,5}(x)-A f_{2,5}(x)-B f_{3,5}(x)-D f_{4,5}(x)}{f_{5}(x)}\right)$and
$F=\left(\frac{f_{1,6}(x)-A f_{2,6}(x)-B f_{3,6}(x)-D f_{4,6}(x)-E f_{5,6}(x)}{f_{6}(x)}\right)$.\\Similarly we get the results for rest of the values of $i$.
\item We have
$\frac{x^n-1}{f_{1}(x)}\left(A_1-A A_2-B A_3-D A_4-E A_5-F A_6-G A_7\right)$\\= 
$uvw\frac{x^n-1}{f_{1}(x)}\left(f_{1,8}(x)-Af_{2,8}(x)-Bf_{3,8}(x)-Df_{4,8}(x)-Ef_{5,8}(x)-Ff_{6,8}(x)-Gf_{7,8}(x)\right)\in C$.
Therefore, $\frac{x^n-1}{f_{1}(x)}\left(f_{1,8}(x)-Af_{2,8}(x)-Bf_{3,8}(x)-Df_{4,8}(x)-Ef_{5,8}(x)-Ff_{6,8}(x)-Gf_{7,8}(x)\right)\in C_8$
$\Rightarrow$ $f_8(x)|\frac{x^n-1}{f_{1}(x)}\left(f_{1,8}(x)-Af_{2,8}(x)-Bf_{3,8}(x)-Df_{4,8}(x)-Ef_{5,8}(x)-Ff_{6,8}(x)-Gf_{7,8}(x)\right)$.
where,
$A=\left(\frac{f_{1,2}(x)}{f_2(x)}\right) $, $B=\left(f_{1,3}(x)-A\frac{f_{2,3}(x)}{f_{3}(x)}\right) $, $ D=\left(\frac{f_{1,4}(x)-Af_{2,4}(x)-Bf_{3,4}(x)}{f_{4}(x)}\right)$,\\
$E=\left(\frac{f_{1,5}(x)-Af_{2,5}(x)-Bf_{3,5}(x)-Df_{4,5}(x)}{f_{5}(x)}\right)$, $F=\left(\frac{f_{1,6}(x)-Af_{2,6}(x)-Bf_{3,6}(x)-Df_{4,6}(x)-Ef_{5,6}(x)}{f_{6}(x)}\right)$and
$G=\left(\frac{f_{1,7}(x)-Af_{2,7}(x)-Bf_{3,7}(x)-Df_{4,7}(x)-Ef_{5,7}(x)-Ff_{6,7}(x)}{f_{7}(x)}\right).$

\item For  $i=4$, we have  
$uA_3$=$uvf_{2,3}(x)+uwf_{2,5}(x)+uvwf_{2,7}(x)\in C.$ Therefore,  $uvf_{2,3}(x)\in  C~mod~w$ $\Rightarrow$ $f_{2,3}(x)\in C_4$ 
$\Rightarrow$ $f_4(x)|f_{2,3}(x)$. Similarly we get the results for rest of the values of $i$.
\item For  $i=4$, we have
 $\left(vA_1-\frac{f_{1}(x)}{f_{3}(x)}A_3\right)\in C$. Since, $uv\left(f_{1,2}(x)-\frac{f_{1}(x)}{f_{3}(x)}f_{3,4}(x)\right)\in C~mod~w$. Therefore,$\left(f_{1,2}(x)-\frac{f_{1}(x)}{f_{3}(x)}f_{3,4}(x)\right)\in C_4$
 $\Rightarrow$ $f_4(x)|\left(f_{1,2}(x)-\frac{f_1(x)}{f_3(x)}f_{3,4}(x)\right)$.\\ Similarly we get the results for rest of the values of $i$.
\item For  $i=7$, we have
$\left( wA_2-\frac{f_{2}(x)}{f_{6}(x)}A_6\right)\in C$. Since, $vw\left(f_{2,3}(x)-\frac{f_{2}(x)}{f_{6}(x)}f_{6,7}(x)\right)\in C~mod~w$. Therefore,
$f_7(x)|\left(f_{2,3}(x)-\frac{f_{2}(x)}{f_{6}(x)}f_{6,7}(x)\right)$. Similarly we get the result for $i=8$.
\item For  $i=7$, we have
$~~\left( wA_1-\frac{f_{1}(x)}{f_{5}(x)}A_5-\frac{\left(f_{1,2}(x)-\frac{f_{1}(x)}{f_{5}(x)}f_{5,6}(x)\right)}{f_{6}(x)}A_6 \right)\in C$. \\Since, 
$vw \left( f_{1,3}(x)-\frac{f_{1}(x)}{f_{5}(x)}f_{5,7}(x)-\frac{\left(f_{1,2}(x)-\frac{f_{1}(x)}{f_{5}(x)}f_{5,6}(x)\right)}{f_{6}(x)}f_{6,7}(x)\right)\in C~mod~uvw$.\\ Therefore, 
$\left( f_{1,3}(x)-\frac{f_{1}(x)}{f_{5}(x)}f_{5,7}(x)-\frac{\left(f_{1,2}(x)-\frac{f_{1}(x)}{f_{5}(x)}f_{5,6}(x)\right)}{f_{6}(x)}f_{6,7}(x)\right)\in C_7$.\\
$\Rightarrow$ $f_7(x)|\left( f_{1,3}(x)-\frac{f_{1}(x)}{f_{5}(x)}f_{5,7}(x)-\frac{\left(f_{1,2}(x)-\frac{f_{1}(x)}{f_{5}(x)}f_{5,6}(x)\right)}{f_{6}(x)}f_{6,7}(x)\right)$.\\ Similarly, we get the results for rest of the values of $i$. 
\item We have
$vA_4 \in C$. Since, $ vwf_{4,5}(x)\in C~mod~uvw $. Therefore,$~~f_{4,5}(x)\in C_7$\\
$\Rightarrow$ $f_7(x)|f_{4,5}(x)$. Similarly by taking $vA_3 $ we can show $f_7(x)|f_{3,5}(x)$.
\item We have,
$uvA_2$=$uvwf_{2,5}(x)\in C$. Therefore, $f_{2,5}(x)\in C_8$ $\Rightarrow$ 
$f_8(x)|f_{2,5}(x)$
\item We have,
$vA_3-\frac{f_{3,5}(x)}{f_7(x)}A_7$= $uvw\left(f_{3,6}(x)-\frac{f_{3,5}(x)}{f_7(x)}f_{7,8}(x)\right)\in C_8$. Therefore,\\  $\left(f_{3,6}(x)-\frac{f_{3,5}(x)}{f_7(x)}f_{7,8}(x)\right)\in C_8$ $\Rightarrow$
$f_8(x)|\left(f_{3,6}(x)-\frac{f_{3,5}(x)}{f_7(x)}f_{7,8}(x)\right)$.
\item We have
$vA_4-\frac{f_{4,5}(x)}{f_7(x)}A_7$=$uvw\left(f_{4,6}(x)-\frac{f_{4,5}(x)}{f_7(x)}f_{7,8}(x)\right)\in C$. Therefore,\\
$\left(f_{4,6}(x)-\frac{f_{4,5}(x)}{f_7(x)}f_{7,8}(x)\right)\in C_8$ $\Rightarrow$
$f_8(x)|\left(f_{4,6}(x)-\frac{f_{4,5}(x)}{f_7(x)}f_{7,8}(x)\right)$.
\item We have
$vA_5-\frac{f_{5}(x)}{f_7(x)}A_7$=$ uvw\left(f_{5,6}(x)-\frac{f_5(x)}{f_7(x)}f_{7,8}(x)\right)\in C$.\\Therefore, $\left(f_{5,6}(x)-\frac{f_5(x)}{f_7(x)}f_{7,8}(x)\right)\in C_8$ 
$\Rightarrow$ $f_8(x)|\left(f_{5,6}(x)-\frac{f_5(x)}{f_7(x)}f_{7,8}(x)\right)$.
\item We have
$wA_1-\frac{f_{1}(x)}{f_5(x)}A_5-A A_6-B A_7$= $uvw\left(f_{1,4}(x)-\frac{f_{1}(x)}{f_5(x)}f_{5,8}(x)-Af_{6,8}(x)-Bf_{7,8}(x)\right)\in C$. Therefore, 
$\left(f_{1,4}(x)-\frac{f_{1}(x)}{f_5(x)}f_{5,8}(x)-Af_{6,8}(x)-Bf_{7,8}(x)\right)\in C_8$\\
$\Rightarrow$  $f_8(x)|\left(f_{1,4}(x)-\frac{f_{1}(x)}{f_5(x)}f_{5,8}(x)-Af_{6,8}(x)-Bf_{7,8}(x)\right)$,\\
where
$A=\left(\frac{f_{1,2}(x)-\frac{f_{1}(x)}{f_5(x)}f_{5,6}(x)}{f_{6}(x)}\right)$ and $B=\left(\frac{f_{1,3}(x)-\frac{f_{1}(x)}{f_5(x)}f_{5,7}(x)-Af_{6,7}(x)}{f_{6}(x)}\right).$
\end{enumerate}
\end{proof}

\begin{theorem} 
If  $C$=$\left\langle A_1, A_2, \cdots, A_8 \right\rangle $ is a cyclic code over the ring $R_{u^2, v^2, w^2, p}$ then $C$ is a free cyclic code if and only if $f_1(x)=f_8(x)$. In this case, we have $C=\left\langle A_1\right\rangle $ and $A_1|(x^n-1)$ in $R_{u^2, v^2, w^2, p}[x]$.
\end{theorem}
\begin{proof}
Let $f_1(x)=f_8(x)$. Since $ f_8(x)|f_4(x)|f_2(x)|f_1(x)$, $f_4(x)|f_3(x)|f_1(x)$,  $f_6(x)|f_2(x)|f_1(x)$, $f_8(x)|f_6(x)|f_5(x)|f_1(x)$, $f_8(x)|f_7(x)|f_3(x)|f_1(x)$ and $f_7(x)|f_5(x)|f_1(x)$. Therefore, we get $f_1(x)=f_2(x)=f_3(x)=f_4(x)=f_5(x)=f_6(x)=f_7(x)=f_8(x)$. Let $B_1 =  f_1(x)+uf_{1,2}(x)+vf_{1,3}(x)+uvf_{1,4}(x), B_2 = uf_2(x)+vf_{2,3}(x)+uvf_{2,4}, B_3 = vf_3(x)+uvf_{3,4}(x), B_{4} =uvf_4(x)$. Then we have Im$\phi=\langle B_1,B_2,B_3,B_4\rangle$ and Ker$\phi=w\langle B_5,B_6,B_7,B_8\rangle$, where $B_5 =  f_5(x)+uf_{5,6}(x)+vf_{5,7}(x)+uvf_{5,8}(x), B_6 = uf_6(x)+vf_{6,7}(x)+uvf_{6,8}, B_7 = vf_7(x)+uvf_{7,8}(x), B_{8} =uvf_8(x)$. From \cite[Proposition 3.3]{KGP15}, we get Im$\phi=\langle B_1\rangle$ and Ker$\phi=w\langle B_5\rangle$. Therefore, we have $C=\langle f_1(x)+uf_{1,2}(x)+vf_{1,3}(x)+uvf_{1,4}(x)+f_{1,5}(x)+uf_{1,6}(x)+vf_{1,7}(x)+uvf_{1,8}(x),wf_5(x)+uwf_{5,6}(x)+vwf_{5,7}(x)+uvwf_{5,8}(x)\rangle $ . Now to show $C=\left\langle A_1\right\rangle,$ we show that $f_{1,2}(x)=f_{5,6}(
x)
$, $f_{1,3}(x)=f_{5,7}(x)$ and $f_{1,4}(x)=f_{5,8}(x)$. Since $wA_1-A_5$=$uw\left(f_{1,2}(x)-f_{5,6}(x)\right)+ vw\left(f_{1,3}(x)-f_{5,7}(x)\right)+uvw\left(f_{1,4}(x)-f_{5,8}(x)\right)\in C.$ This gives, $\left(f_{1,2}(x)-f_{5,6}(x)\right)\in C_6=\left\langle f_6(x)\right\rangle$. Therefore, $f_6(x)|\left(f_{1,2}(x)-f_{5,6}(x)\right)$. Since $({\rm deg}(f_{1,2}(x)),{\rm deg}(f_{5,6}(x)))<{\rm deg}(f_6(x))$. This implies that $f_{1,2}(x)-f_{5,6}(x)=0$. Thus, $f_{1,2}(x)=f_{5,6}(x)$.  Hence, we get $wA_1-A_5=vw\left(f_{1,3}(x)-f_{5,7}(x)\right)+uvw\left(f_{1,4}(x)-f_{5,8}(x)\right)$ $\in C.$ This gives, $\left(f_{1,3}(x)-f_{5,7}(x)\right)\in C_7=\left\langle f_7(x)\right\rangle$. Therefore, $f_7(x)|\left(f_{1,3}(x)-f_{5,7}(x)\right)$. Since $({\rm deg}(f_{1,3}(x)),{\rm deg}(f_{5,7}(x)))<{\rm deg}(f_7(x))$. This implies that $\left(f_{1,3}(x)-f_{5,7}(x)\right)=0$. Therefore, $f_{1,3}(x)=f_{5,7}(x)$. Finally, we have $wA_1-A_5=uvw\left(f_{1,4}(x)-f_{5,8}(x)\right)\in C.$ This gives, $\left(f_{1,4}(x)-f_{5,8}(x)
\right)\in C_8=\left\langle f_8(x)\right\rangle$. Therefore, $f_8(x)|\left(f_{1,4}(x)-f_{5,8}(x)\right)$. Since $({\rm deg}(f_{1,4}(x)),{\rm deg}(f_{5,8}(x)))<{\rm deg}(f_8(x))$. This implies that $\left(f_{1,4}(x)-f_{5,8}(x)\right)=0$. Therefore, $f_{1,4}(x)=f_{5,8}(x)$. This shows that $wA_1=A_5$. Hence $C=\left\langle A_1\right\rangle $ and $C \simeq R_{u^2, v^2, w^2, p}^{n-{\rm deg}(f_1(x))}$. Conversely, if $C$ is a free cyclic code, we must have $C =\langle A_1 \rangle$. Since $uvw f_8(x) \in C$, we have $uvwf_8(x) = 
uvw \alpha f_1(x)$ for some $\
\alpha \in \F_p$. Note that $f_8(x)|f_1(x)$, hence by comparing the coefficients both sides, we get $f_1(x) = f_8(x)$. Now for the second condition, by division algorithm, we have $x^n-1=A_1q(x)+r(x)$, where $r(x)=0$ or ${\rm deg}(r(x))<{\rm deg}(f_1(x)).$ This implies that $r(x)=(x^n-1)-A_1q(x)\in C$. Since $A_1$ is the lowest degree polynomial 
in $C$. So $r(x)=0$. Hence, $A_1|(x^n-1)$ in $R_{u^2, v^2, w^2, p}[x]$.
\end{proof}

\subsection{\rm When $n$ is relatively prime to $p$} 
Let $n$ be a positive integer relatively prime to $p$. First, we slightly refine Theorem 3.4 of \cite{KGP15}, which gives the structure of a cyclic code over the ring $R_{u^2,v^2,p}=\F_p[u,v]/\langle u^2, v^2, uv-vu \rangle$. Let $C_{u,v}$ be a cyclic code over the ring $R_{u^2,v^2,p}$. From Theorem 3.4 of \cite{KGP15}, we have $C_{u,v}=\langle g(x)+ua_1(x)+uvr_1(x), va_2(x)+uva_3(x)\rangle$ with $a_{3}|a_1(x)|g(x)|(x^n-1)$ and $ a_3(x)|a_2(x)|g(x)|(x^n-1).$ Thus, $\frac{x^n-1}{g(x)}\frac{x^n-1}{a_1(x)}(g(x)+ua_1(x)+uvr_1(x)) = uv\frac{x^n-1}{g(x)}\frac{x^n-1}{a_1(x)}r_1(x) \in C_{u,v}$. This gives, $\frac{x^n-1}{g(x)}\frac{x^n-1}{a_1(x)}r_1(x) \in C_4=\text{Tor}(\text{Tor}(C_{u,v})) = \langle a_3(x)\rangle$ (see Page 165 of \cite{KGP15}). Hence, $a_3(x)|\frac{x^n-1}{g(x)}\frac{x^n-1}{a_1(x)}r_1(x).$ Since $n$ is relatively prime to $p$, $x^n-1$ can be uniquely factored as product of distinct irreducible factors. Therefore, we must have $\text{g.c.d.}\left(a_{3}(x), \frac{x^n-1}{g(x)}\right) = \text{g.c.d.}
\left(a_{3}(x), \frac{x^n-1}{a_1(x)}\right)= 1.$ This gives, $a_{3}(x)|r_1(x)$. But, from Theorem 3.1 of \cite{KGP15}, we have $\text{deg}(r_{1}(x))<\text{deg}(a_{3}(x))$. This gives, $r_{1}(x)=0$. Thus we have proved the following theorem.
\begin{theo}
Let $C_{u,v}$ be a cyclic code over the ring $R_{u^2,v^2,p}$ of length $n$. If $n$ is relatively prime to $p$, then we have $C_{u,v}=\langle f_1(x)+uf_2(x), vf_3(x)+uvf_4(x)\rangle$ with $f_4(x)|f_2(x)|f_1(x)|(x^n-1)$ and $ f_4(x)|f_3(x)|f_1(x)|(x^n-1).$
\end{theo}

If $C=\langle A_1, A_2, \cdots, A_{8}\rangle$ is a cyclic code of length $n$ over the ring $R_{u^2,v^2,w^2,p}$ then we have $\text{Im}\phi=\langle A_1, A_2, A_3, A_4 \rangle$ and $\text{ker}\phi=w\langle A_5, A_6, A_7, A_8\rangle$. (See Equation \ref{surj-hom} for the definition of $\phi$). Note that we can consider $\text{Im}\phi$ and  $\text{ker}\phi$ as cyclic codes over the ring $R_{u^2,v^2,p}$. Since $n$ is relatively prime to $p$, from the above theorem, we have $\text{Im}\phi=\langle f_1(x)+uf_2(x), vf_3(x)+uvf_4(x)\rangle$ and $\text{ker}\phi=w\langle f_5(x)+uf_6(x), vf_7(x)+uvf_8(x)\rangle$ with $f_4(x)|f_2(x)|f_1(x)|(x^n-1)$, $ f_4(x)|f_3(x)|f_1(x)|(x^n-1),~$ $f_8(x)|f_6(x)|f_5(x)|(x^n-1)$ and $f_8(x)|f_7(x)|f_5(x)|(x^n-1).$ We also have the conditions $f_5(x)|f_1(x), f_6(x)|f_2(x)$ and $f_7(x)|f_3(x)$ (from Conditions 1 and 2 of Theorem \ref{properties}). Therefore, the code $C$ can be written as $C=\langle f_1(x)+uf_2(x)+w(f_{1,5}(x)+uf_{1,6}(x)+vf_{1,7}(x)+uvf_{1,8}(x)), f_3(x)+uf_4(x)+w(f_{3,5}(
x)
+uf_{3,6}(x)+vf_{3,7}(x)+uvf_{3,8}(x)), w(f_5(x)+uf_6(x)), w(vf_7(x)+uvf_8(x))\rangle$ with the same conditions as above on $f_i(x)$'s. From Condition 4 of Theorem \ref{properties}, for $i=1$ and $4 \leq j \leq 7$, we get $f_{1+j}(x)|\frac{x^n-1}{f_1(x)}\frac{x^n-1}{f_2(x)} \cdots \frac{x^n-1}{f_{j}(x)}f_{1,1+j}(x)$. Since $n$ is relatively prime to $p$, $x^n-1$ can be uniquely factored as product of distinct irreducible factors. Therefore, we must have $\text{g.c.d.}\left(f_{1+j}(x), \frac{x^n-1}{f_k(x)}\right) = 1$, for $1\leq k \leq j$. This gives $f_{1+j}(x)|f_{1,1+j}(x)$. From Theorem \ref{unique}, we have $\text{deg}(f_{1,j+1}(x))<\text{deg}(f_{1+j}(x))$, for $4\leq j \leq 7$. This gives $f_{1,1+j}(x)=0$, for $4 \leq j \leq 7$. Similarly, from Condition 4 of Theorem \ref{properties}, for $i=3$ and $2 \leq j \leq 5$, we can show that $f_{3,3+j}(x)=0.$ Thus we have proved the following theorem.

\begin{theo}\label{relatively-prime}
Let $C=\langle A_1, A_2, \cdots, A_{8}\rangle$ be a cyclic code over the ring $R_{u^2,v^2, w^2, p}$ of length $n$. If $n$ is relatively prime to $p$, then we have $C=\langle f_1(x)+uf_2(x), vf_3(x)+uvf_4(x), w(f_5(x)+uf_6(x)), w(vf_7(x)+uvf_8(x))\rangle$ with the conditions $f_4(x)|f_2(x)|f_1(x)|(x^n-1)$, $ f_4(x)|f_3(x)|f_1(x),$ $f_8(x)|f_6(x)|f_5(x)|(x^n-1)$, $f_8(x)|f_7(x)|f_5(x)|f_1(x), f_6(x)|f_2(x)$ and $f_7(x)|f_3(x)$.
\end{theo}

\section{Ranks and minimal spanning sets}
We follow Dougherty and Shiromoto \cite[page 401]{Dou-Shiro01} for the definition of the rank of a code $C$. We first prove the number of lemmas that we use to find the rank and the minimal spanning set of cyclic codes over $R_{u^2, v^2, w^2, p}$.
\begin{lemma} \label{rak1}
Let $C$ be a cyclic code over the ring $R_{u^2, v^2, w^2, p}$. If $C=\langle A_1, A_2, \cdots, A_8\rangle$  then polynomials in $C$ in the following  forms can be written as follows:
\begin{enumerate}[{\rm (1)}]
\item  $w(p_0(x)+up_1(x)+vp_2(x)+uvp_3(x))$ = $q_5(x)A_5+q_6(x)A_6+q_7(x)A_7+q_8(x)A_8$,
\item  $w(up_1(x)+vp_2(x)+uvp_3(x))$ = $q_6(x)A_6+q_7(x)A_7+q_8(x)A_8$,
\item  $w(vp_2(x)+uvp_3(x))$ = $q_7(x)A_7+q_8(x)A_8$,
\item  $w(uvp_3(x))$ = $q_8(x)A_8$,
\end{enumerate}
for some $q_i(x)\in \F_p[x]$, $5\leq i \leq 8$.
\end{lemma}
\begin{proof}
$(1)$ Let $A'=w(p_0(x)+up_1(x)+vp_2(x)+uvp_3(x))\in C$. Thus,  $p_0(x)\in C_5= \langle f_5(x) \rangle$. This gives, $p_0(x)=q_5(x)f(x)$ for some $q_5(x) \in \F_p[x]$.
Therefore,
$A'-q_5(x)A_5=w\left((p_1(x)-q_5(x)f_{5,6}(x))+u(p_2(x)-q_5(x)f_{5,7}(x))+uv(p_3(x)-q_5(x)f_{5,8}(x))\right)$ $\in C$.
Thus,
      $p_1(x)-q_5(x)f_{5,6}(x) \in C_6= \langle f_6(x) \rangle.\notag$
Therefore, $(p_1(x)-q_5(x)f_{5,6}(x))$=$q_6(x)f_6(x)$ for some $q_6(x) \in \F_p[x]$.
Again,
\begin{multline}
A'-q_5(x)A_5-q_6(x)A_6=w(v(p_2(x)-q_5(x)f_{5,7}(x)-q_1(x)f_{6,7}(x))\\+uv(p_3(x)-q_5(x)f_{5,8}(x)-q_1(x)f_{6,8}(x)))\in C.\notag
\end{multline}
Thus,
      $(p_2(x)-q_5(x)f_{5,7}(x)-q_1(x)f_{6,7}(x)) \in C_7= \langle f_7(x) \rangle.\notag$
Therefore,
$(p_2(x)-q_5(x)f_{5,7}(x)-q_6(x)f_{6,7}(x))$=$q_7(x)f_7(x)$ for some $ q_7(x) \in \F_p[x]$.
Again,
\begin{multline}
A'-q_5(x)A_5-q_6(x)A_6-q_7(x)A_7\\=w(uv(p_3(x)-q_5(x)f_{5,8}(x)-q_6(x)f_{6,8}(x)-q_7(x)f_{7,8}(x))\in C.\notag 
\end{multline}
Thus,
      $(p_3(x)-q_5(x)f_{5,8}(x)-q_6(x)f_{6,8}(x)-q_7(x)f_{7,8}(x)) \in C_8= \langle f_8(x) \rangle.\notag$
Therefore,
$(p_3(x)-q_5(x)f_{5,8}(x)-q_6(x)f_{6,8}(x)-q_7(x)f_{7,8}(x))$=$q_8(x)f_8(x)$
for some $ q_8(x) \in \F_p[x]$. That is,
$A'-q_5(x)A_5-q_6(x)A_6-q_7(x)A_7-q_8(x)A_8=0$
$\Rightarrow $
$A'=q_5(x)A_5+q_6(x)A_6+q_7(x)A_7+q_8(x)A_8.\notag$
This proves Statement ($1$). The proof of other cases  are similar to the proof of Statement ($1$).
\end{proof}

\begin{lemma} \label{lem2}
Let $C$ be a cyclic code over the ring $R_{u^2, v^2, w^2, p}$. If $C=\langle A_1, A_2, \cdots, A_8\rangle$ and ${\rm deg}(f_{i}(x)) = t_i, 1\leq i \leq 8,$ then the following conditions hold:
\begin{enumerate}[{\rm (1)}]
\item $x^{t_i-t_8}A_8$ = $ c_iu_{8-i}A_i+q_8(x)A_8$, $ 1 \leq i \leq 7$, where ${\rm deg}(q_8(x)) < t_i-t_8 $, $u_1 = u, u_2 = v, u_3 = uv, u_4 = w, u_5 = uw, u_6 = vw ~\text{and}~ u_7 = uvw$,\label{cond1-l2}
\item $x^{t_5-t_7}A_7$ = $c_5vA_5-q_7'(x)A_7-q_8'(x)A_8$, where ${\rm deg}(q_7'(x)) < t_5-t_7 $, \label{cond2-l2}
\item $x^{t_3-t_7}A_7$=$c_3wA_3-q_7'(x)A_7-q_8'(x)A_8$, where ${\rm deg}(q_7'(x)) < t_3-t_7 $,\label{cond3-l2}
\item $x^{t_1-t_7}A_7$=$c_1vwA_1-q_7'(x)A_7-q_8'(x)A_8$, where ${\rm deg}(q_7'(x)) < t_1-t_7 $,\label{cond4-l2}
\item $x^{t_5-t_6}A_6=c_5uA_5-q_6'(x)A_6-q_7'(x)A_7-q_8'(x)A_8$, where ${\rm deg}(q_6'(x))<t_5-t_6$,\label{cond5-l2}
\item $x^{t_2-t_6}A_6=c_2wA_2-q_6'(x)A_6-q_7'(x)A_7-q_8'(x)A_8$, where ${\rm deg}(q_6'(x))<t_2-t_6$,\label{cond7-l2}
\item $x^{t_1-t_6}A_6=c_1uwA_1-q_6'(x)A_6-q_7'(x)A_7-q_8'(x)A_8$, where ${\rm deg}(q_6'(x))<t_1-t_6$, and \label{cond6-l2}
\item $x^{t_1-t_5}A_5=c_1wA_1+q_5(x)A_5+q_6(x)A_6+q_7(x)A_7+q_8(x)A_8$, where ${\rm deg}(q_5(x))<t_1-t_5$,\label{cond8-l2}
\end{enumerate}
$c_i \in \F_p$ and $q_i(x), q_i'(x) \in \F_p[x]$.
\end{lemma}
\begin{proof}
\begin{enumerate}
 \item  From Condition (1) of Theorem 4.2, we have
$f_8(x)|f_i(x)$, $1 \leq i \leq 7$. Thus, $f_i(x)=s_i(x)f_8(x)$ for some $s_i(x) \in \F_p[x] $. This can be written as,
$f_i(x)=(s_{i0}+xs_{i1}+ \cdots+x^{t_i-t_8}s_{i(t_i-t_8)})f_8(x) $, where $ s_{ij} \in \F_p$. Clearly $s_{i(t_i-t_8)}\neq 0.$
Therefore, $u_{8-i}A_i-s_i(x)A_8 = uvw(f_i(x)-s_i(x)f_8(x))=0.$ This gives, $x^{t_i-t_8}A_8=s_{i(t_i-t_8)}^{-1}u_{8-i}A_i-s_{i(t_i-t_8)}^{-1}(s_{i0}+xs_{i1}+ \cdots+x^{t_i-t_8-1}s_{i(t_i-t_8-1)})A_8$. Hence, $x^{t_i-t_8}A_8=c_iu_{8-i}A_i+q_8(x)A_8$, where ${\rm deg}(q_8(x))<t_i-t_8$.
\item From Condition (2) of Theorem 4.2, we have $f_7(x)|f_5(x)$. Thus, $f_5(x)=s_5(x)f_7(x)$ for some $s_5(x) \in F_p[x]$. This can be written as $f_5(x)=(s_{50}+xs_{51}+ \cdots+x^{t_5-t_7}s_{5(t_5-t_7)})f_7(x)$.  This together with Condition (4) of Lemma \ref{rak1}, we get  $vA_5-s_5(x)A_7=w(uv(f_{5,6}(x)-s_5(x)f_{7,8}(x))) = q_8(x)A_8 $. Thus,
          \begin{equation} \label{l31}
  s_5(x)A_7(x)=vA_5-q_8(x)A_8.
  \end{equation}
 This can be written as 
$x^{t_5-t_7}A_7=s_{5(t_5-t_7)}^{-1}vA_5-s_{5(t_5-t_7)}^{-1}(s_{50}+xs_{51}+ \cdots +$ $x^{t_5-t_7-1}$ $s_{5(t_5-t_7-1)})A_7(x)-s_{5(t_5-t_7)}^{-1}q_8(x)A_8$. Thus,
\begin{equation} \label{l32}
 x^{t_5-t_7}A_7=c_5vA_5-q_7'(x)A_7-q_8'(x)A_8,
\end{equation}
 where ${\rm deg}(q_7'(x))=t_5-t_7-1<t_5-t_7$.

   \item The proof is similar to Condition \ref{cond2-l2}.
    \item The proof is similar to Condition \ref{cond2-l2}.
\item  From Condition (2) of Theorem 4.2, we have
    $f_6(x)|f_5(x)$. Thus, $f_5(x)=s_5(x)f_6(x)$ for some $s_5(x) \in F_p[x]$. This can be written as,
    $f_5(x)=(s_{50}+s_{51}x+ \cdots+s_{5(t_5-t_6)}x^{t_5-t_6})f_6(x)$, where $ s_{5i} \in \F_p$. This together with Condition (3) of Lemma \ref{rak1}, we get $uA_5-s_5(x)A_6(x)=w(v(-s_5(x)f_{6,7}(x))+uv(f_{5,7}(x)-s_5(x)f_{6,8}(x)))=q_7(x)A_7+q_8(x)A_8\in C.$ Thus, 
    \begin{equation} \label{l26}
    s_5(x)A_6(x)=uA_5-q_7(x)A_7-q_8(x)A_8. 
    \end{equation}
    This can be written as
  $x^{t_5-t_6}A_6=s_{5(t_5-t_6)}^{-1}uA_5-s_{5(t_5-t_6)}^{-1}(s_{50}+xs_{51}+ \cdots+x^{t_5-t_6-1}$ $s_{5(t_5-t_6-1)})A_6-s_{5(t_5-t_6)}^{-1}q_7(x)A_7-s_{5(t_5-t_6)}^{-1}q_8(x)A_8$, i.e.,
\begin{equation}
 x^{t_5-t_6}A_6 =c_5uA_5-q_6'(x)A_6-q_7'(x)A_7-q_8'(x)A_8,
\end{equation}
       where ${\rm deg}(q_6'(x))=t_5-t_6-1<t_5-t_6.$ 
    \item The proof is similar to Condition \ref{cond5-l2}.  
       \item The proof is similar to Condition \ref{cond5-l2}.  
       
     \item By the division algorithm, we have
   \begin{multline}\label{eqt15}
     x^{t_1-t_5}(f_5(x)+uf_{5,6}(x)+vf_{5,7}(x)+uvf_{5,8}(x))=c_1A_1+(p_0(x)+up_1(x)+vp_2(x)\\+uvp_3(x)+wp_4(x)+uwp_5(x)+vwp_6(x)+uvwp_7(x))
     \end{multline}
    where ${\rm deg}(p_0(x))<{\rm deg}(f_1(x))=t_1$.   
    Multiplying Eq.(\ref{eqt15}) by $w$ and applying Condition 1 of Lemma \ref{rak1} gives, 
     $ x^{t_1-t_5}A_5-c_1wA_1=w(p_0(x)+up_1(x)+vp_2(x)+uvp_3(x))=q_5(x)A_5+q_6(x)A_6+q_7(x)A_7+q_8(x)A_8$.
          That is,
          \begin{equation} \label{l33}
           x^{t_1-t_5}A_5=c_1wA_1+q_5(x)A_5+q_6(x)A_6+q_7(x)A_7+q_8(x)A_8.
           \end{equation}
          We have $ p_0(x)=q_5(x)f_5(x)$, thus,    
${\rm deg}(q_5(x))$+ ${\rm deg}(f_5(x)) = {\rm deg}(p_0(x))< t_1$. Hence,     
${\rm deg}(q_5(x)) < t_1-t_5 $. 
   \end{enumerate}
\end{proof}
\begin{theo} \label{rank-main}
Let $C$ be a cyclic code of length $n$ over $R_{u^2, v^2, w^2, p}$. If  $C = \langle A_1, A_2, A_3, A_4,$ $A_5, A_6, A_7, A_8\rangle$ with $t_i={\rm deg}(f_i(x))$, $1 \leq i \leq 8$, $ t_4'={\rm min} \{t_2,t_3\}$, $ t_6'={\rm min}\{t_2,t_5\}$, $ t_7'={\rm min}\{t_3,t_5\} $ and $ t_8'={\rm min}\{t_4,t_6,t_7\}$, then $ C $ has rank $ n+2t_1+t_4'+t_6'+t_7'+t_8'-t_2-t_3-t_4-t_5-t_6-t_7-t_8$. The minimal spanning set $B$ of the code $ C $ is  $ B=\{A_1, xA_1, $ $\cdots, x^{n-t_1-1}A_1, A_2, xA_2,\cdots,$ $ x^{t_1-t_2-1}A_2, $ $A_3, xA_3, $ $\cdots, x^{t_1-t_3-1}A_3, A_4, xA_4, $ $\cdots, x^{t_4'-t_4-1}$ $A_4
 ,A_5, xA_5, $ $\cdots, x^{t_1-t_5-1}A_5, A_6, xA_6,\cdots, x^{t_6'-t_6-1}A_6,A_7, xA_7,\cdots, x^{t_7'-t_7-1}A_7, A_8, xA_8,$\\
 $\cdots, x^{t_8'-t_8-1}A_8\}$. 
 \end{theo}
 \begin{proof}
 It is suffices to show that $B$ spans the set $B'=\{A_1,xA_1,\cdots, x^{n-t_1-1}A_1,A_2,xA_2,\cdots, \\x^{n-t_2-1}A_2,A_3,xA_3,\cdots, x^{n-t_3-1}A_3,A_4,xA_4, \cdots, x^{n-t_4-1}A_4,A_5,xA_5,\cdots, x^{n-t_5-1}A_5,A_6,\\xA_6,\cdots, x^{n-t_6-1}A_6,A_7,xA_7,\cdots, x^{n-t_7-1}A_7,A_8,xA_8,\cdots, x^{n-t_8-1}A_8\} $. To show $B$ spans $B'$, we write the set $B'$ as $B'=B_1\cup B_2$, where $B_1=\{A_1,xA_1,\cdots, x^{n-t_1-1}A_1,A_2,xA_2,\cdots, \\x^{n-t_2-1}A_2,A_3,xA_3,\cdots, x^{n-t_3-1}A_3,A_4,xA_4, \cdots, x^{n-t_4-1}A_4\}$ and $B_2=\{A_5,xA_5,\cdots,\\ x^{n-t_5-1}A_5,A_6,xA_6,\cdots, x^{n-t_6-1}A_6,A_7,xA_7,\cdots, x^{n-t_7-1}A_7,A_8,xA_8,\cdots, x^{n-t_8-1}A_8\}$. First we show that $B$ spans $B_2$ and then we show that $B$ spans $B_1$. To show $B$ spans $B_2$ we divide the proof in twelve cases.\\
{\bf Case (1): } Let $t_8'=t_7$, $t_7'=t_5$ and $t_6'=t_5$. We first show that the element $x^{t_7-t_8}A_8\in B_2-B $ is linear combinations of some elements of B and then we show that other elements of the set $B_2-B$ are linear combinations of elements of B. From Statement \ref{cond1-l2} of Lemma \ref{lem2},
\begin{equation}\label{eq1}
x^{t_7-t_8}A_8= c_7uA_7+q_8(x)A_8, 
\end{equation}
where ${\rm deg}(q_8(x))<t_7-t_8$. Therefore, $x^{t_7-t_8}A_8 \in \text{Span}(B)$. Multiplying Equation (\ref{eq1}) by $x,x^2, x^3, \cdots, x^{t_5-t_7-1}$ and then putting the value of $x^{t_7-t_8}A_8$ in the equation obtained, we can show that $x^{t_7-t_8+1}A_8,x^{t_7-t_8+2}A_8, \cdots, x^{t_5-t_8-1}A_8 \in \text{Span}(B)$. From Statement \ref{cond1-l2} of Lemma \ref{lem2}, we have 
\begin{equation}\label{eq11}
x^{t_5-t_8}A_8= c_5uvA_5+q_8(x)A_8 ,
\end{equation}
where ${\rm deg}(q_8(x))<t_5-t_8$. 
  Therefore, $x^{t_5-t_8}A_8 \in \text{Span}(B)$. Arguing as above, we can show that the terms $x^{t_5-t_8+1}A_8$,$x^{t_5-t_8+2}A_8$,$\cdots$, $x^{t_1-t_8-1}A_8\in \text{Span}(B)$. Again, from Statement \ref{cond1-l2} of Lemma \ref{lem2}, we have 
\begin{equation}\label{eq11r}
x^{t_1-t_8}A_8= c_5uvwA_1+q_8(x)A_8, 
\end{equation}
where ${\rm deg}(q_8(x))<t_1-t_8$. As above, we can show that $x^{t_1-t_8}A_8$, $x^{t_1-t_8+1}A_8$, $\cdots$, $x^{n-t_8-1}A_8$ $\in \text{Span}(B)$.
Now we show that $x^{t_5-t_7}A_7 \in \text{Span}(B)$. From Statement \ref{cond2-l2} of Lemma \ref{lem2}, we have 
\begin{equation} \label{eq14}
x^{t_5-t_7}A_7=c_5vA_5-q_7'(x)A_7-q_8'(x)A_8,
\end{equation}
where $\text{deg}(q_7'(x))<(t_5-t_7)$. In the above discussion, we have shown that $x^iA_8 \in \text{Span}(B)$, for $1 \leq i \leq n-t_8-1$. Clearly, $q_8'(x)A_8 \in \text{Span}(B)$. And, also the term $c_5vA_5$, $q_7'(x)A_7 \in \text{Span}(B)$ (since $\text{deg}(q_7(x))<(t_5-t_7)$). Therefore, $x^{t_5-t_7}A_7 \in \text{Span}(B)$.
 As above, after putting the value of $x^{t_5-t_7}A_7$ in the equation obtained by multiplying Equation (\ref{eq14}) by $x, x^2, \cdots, x^{t_1-t_5-1}$, successively, we can show that $x^{t_5-t_7+1}A_7$, $x^{t_5-t_7+2}A_7,$  $\cdots$, $x^{t_1-t_7-1}A_7 \in \text{Span}(B)$. From Statement \ref{cond4-l2} of Lemma \ref{lem2}, we have
\begin{equation} \label{eq20}
 x^{t_1-t_7}A_7=c_1vwA_1-q_7'(x)A_7-q_8'(x)A_8,
\end{equation}
where $\text{deg}(q_7'(x))<t_1-t_7$. As above, we can show that $x^{t_1-t_7}A_7$, $x^{t_1-t_7+1}A_7$, $\cdots$, $x^{n-t_7-1}A_7 \in \text{Span}(B)$. Now we show that $x^{t_5-t_6}A_6 \in \text{Span}(B)$. From Statement \ref{cond5-l2} of Lemma \ref{lem2}, we have 
\begin{equation} \label{eq17}
x^{t_5-t_6}A_6=c_5uA_5-q_6'(x)A_6-q_7'(x)A_7-q_8'(x)A_8,
\end{equation}
where $\text{deg}(q_6'(x))<(t_5-t_6)$. In  the above discussion, we have shown that $x^iA_j \in \text{Span}(B)$, for $1 \leq i \leq n-t_j-1$, $j=7,8$.  Clearly, $q_7'(x),q_8'(x)A_8 \in \text{Span}(B)$. Therefore, $x^{t_5-t_6}A_6 \in \text{Span}(B)$. In a similar way, as above, after putting the value of $x^{t_5-t_6}A_6$ in the equation obtained by multiplying Equation (\ref{eq17}) by $x,x^2, x^3, \cdots, x^{t_1-t_5-1}$, successively, we can show that $x^{t_5-t_6+1}A_6, x^{t_5-t_6+2}A_6$, $\cdots$, $x^{t_1-t_6-1}A_6 \in \text{Span}(B)$. From Statement \ref{cond6-l2} of Lemma \ref{lem2}, we have 
\begin{equation} \label{eq201}
 x^{t_1-t_6}A_6=c_1uwA_1-q_6'(x)A_6-q_7'(x)A_7-q_8'(x)A_8,
\end{equation}
where $\text{deg}(q_6'(x))<t_1-t_6$.  Again as above, we can show that $x^{t_1-t_6+1}A_6$, $x^{t_1-t_6+2}A_{6}$, $\cdots$, $x^{n-t_6-1}A_6 \in \text{Span}(B)$. Now we show that the next term $x^{t_1-t_5}A_5 \in \text{Span}(B)$. From Statement \ref{cond8-l2}  of Lemma \ref{lem2}, we have 
\begin{equation} \label{eqO16}
x^{t_1-t_5}A_5=c_1wA_1+q_5(x)A_5+q_6(x)A_6+q_7(x)A_7+q_8(x)A_8,
\end{equation}
where $\text{deg}(q_5(x))<t_1-t_5.$ In the above discussion, we have shown that $x^iA_j \in \text{Span}(B)$, for $1 \leq i \leq n-t_j-1$, $j=6,7,8$.  Clearly, $q_6(x)A_6,q_7(x),q_8(x)A_8 \in \text{Span}(B)$. Therefore, $x^{t_1-t_5}A_5 \in \text{Span}(B)$ (since $\text{deg}(q_5(x))<t_1-t_5)$. Multiplying Equation (\ref{eqO16}) by $x,x^2, x^3, \cdots, x^{n-t_1-1}$ and then putting the value of  $x^{t_1-t_5}A_5$ in the equation obtained, we can show that the terms $x^{t_1-t_5+1}A_5$,$x^{t_1-t_5+1}A_5$,$\cdots$, $x^{n-t_5-1}A_5\in \text{Span}(B)$. 
 
{\bf Case (2A):} Let $t_8'=t_4$, $t_7'=t_3$ and  $t_6'=t_2$. Let $t_4'=t_3$. As in Case 1, by using Statement \ref{cond1-l2} of Lemma \ref{lem2} for $i = 4, 3$ and $1$, successively, we can show that $x^{t_4-t_8}A_8,x^{t_4-t_8+1}A_8,\cdots,x^{n-t_8-1}A_8 \in \text{Span}(B)$. Similarly, as in Case 1, by using Statements \ref{cond3-l2} and \ref{cond4-l2} of Lemma \ref{lem2}, successively, we can show that $x^{t_3-t_7}A_7,$ $x^{t_3-t_7+1}A_7,\cdots,$ $x^{n-t_7-1}A_7 \in \text{Span}(B)$. Again, as in Case 1, by using Statement \ref{cond7-l2} and then Statement \ref{cond6-l2} of Lemma \ref{lem2}, successively, we can show that $x^{t_2-t_6}A_6,x^{t_2-t_6+1}A_6,\cdots,x^{n-t_6-1}A_6 \in \text{Span}(B)$. In a similar 
fashion, as in Case 1, by using Statement \ref{cond8-l2} of Lemma \ref{lem2}, we can show that $x^{t_1-t_5}A_5,x^{t_1-t_5+1}A_5,$ $\cdots,x^{n-t_5-1}A_5 \in \text{Span}(B). $

{\bf Case (2B):} Let $t_8'=t_4$, $t_7'=t_3$, $t_6'=t_2$ and $t_4'=t_2$. As in Case 1, by using Statement \ref{cond1-l2} of Lemma \ref{lem2} for $i = 4, 2$ and $1$, successively, we can show that $x^{t_4-t_8}A_8,x^{t_4-t_8+1}A_8,\cdots,x^{n-t_8-1}A_8 \in \text{Span}(B)$. Similarly, as in Case 1, by using Statements \ref{cond3-l2}, \ref{cond4-l2},  \ref{cond7-l2}, \ref{cond6-l2} and then \ref{cond8-l2} of Lemma \ref{lem2}, successively, we can show that the rest of elements belongs to $\text{Span}(B)$.

{\bf Case (3): } Let $t_8'=t_6$, $t_7'=t_3$ and $t_6'=t_5$. As in Case 1, by using Statement \ref{cond1-l2} of Lemma \ref{lem2} for $i = 6, 5$ and $1$, successively, we can show that $x^{t_6-t_8}A_8,x^{t_6-t_8+1}A_8,\cdots,$ $x^{n-t_8-1}A_8 \in \text{Span}(B)$. Similarly, as in Case 1, by using Statements \ref{cond3-l2} and \ref{cond4-l2} of Lemma \ref{lem2}, successively, we can show that $x^{t_3-t_7}A_7,$ $x^{t_3-t_7+1}A_7,\cdots,$ $x^{n-t_7-1}A_7 \in \text{Span}(B)$. Again, as in Case 1, by using Statement \ref{cond5-l2} and then Statement \ref{cond6-l2} of Lemma \ref{lem2}, successively, we can show that $x^{t_5-t_6}A_6,x^{t_5-t_6+1}A_6,\cdots,x^{n-t_6-1}A_6 \in \text{Span}(B)$. In a similar 
fashion, as in Case 1, by using Statement \ref{cond8-l2} of Lemma \ref{lem2}, we can show that $x^{t_1-t_5}A_5,x^{t_1-t_5+1}A_5,$ $\cdots,x^{n-t_5-1}A_5 \in \text{Span}(B). $

The remaining cases are as follows: 
{\bf Case (4):} If $t_8'=t_7$, $t_7'=t_3$ and $t_6'=t_5$.
{\bf Case (5):} If $t_8'=t_7$, $t_7'=t_5$ and $t_6'=t_2$.
{\bf Case (6):} If $t_8'=t_7$, $t_7'=t_5$ and $t_6'=t_5$.
{\bf Case (7):} If $t_8'=t_6$, $t_7'=t_3$ and $t_6'=t_2$.
{\bf Case (8):} If $t_8'=t_6$, $t_7'=t_5$ and $t_6'=t_2$.
{\bf Case (9):} If $t_8'=t_6$, $t_7'=t_5$ and $t_6'=t_5$.
{\bf Case (10):} If $t_8'=t_4$, $t_7'=t_5$ and $t_6'=t_5$,{\bf(10A):}  $t_4'=t_3$,{\bf (10B):}  $t_4'=t_2$.
{\bf Case (11):} If $t_8'=t_4$, $t_7'=t_3$ and $t_6'=t_5$,{\bf(11A):}  $t_4'=t_3$,{\bf (11B):}  $t_4'=t_2$.
{\bf Case (12):} If $t_8'=t_4$, $t_7'=t_5$ and $t_6'=t_2$,{\bf(12A):}  $t_4'=t_3$,{\bf (12B):}  $t_4'=t_2$.
In a similar way as above, by using Statement Lemma \ref{lem2}, we can show that $B$ spans $B_2$ in these cases.

Now we show that $B$ spans $B_1$. From Equation (\ref{surj-hom}), we have a homomorphism $\phi : C \rightarrow R_{u^2,v^2,p,n}$. Therefore, $C/\text{Ker}\phi \simeq \phi(C)$ and $\phi(C)$ is a cyclic code over the ring $R_{u^2,v^2,p}$. Thus, we have $C/\text{Ker}\phi$ as a cyclic code over $R_{u^2,v^2,p}$. Therefore, from Theorem 4.1 of \cite{KGP15}, the minimal spanning set $B_{\phi}$ of the code $C/\text{Ker}\phi$ is $\{A_1+\text{Ker}\phi, xA_1+\text{Ker}\phi$, $\cdots$, $x^{n-t_1-1}A_1+\text{Ker}\phi$, $A_2+\text{Ker}\phi, xA_2+\text{Ker}\phi$, $\cdots$, $x^{t_1-t_2-1}A_2+\text{Ker}\phi, A_3+\text{Ker}\phi, xA_3+\text{Ker}\phi$, $\cdots$, $x^{t_1-t_3-1}A_3+\text{Ker}\phi$, $A_4+\text{Ker}\phi, xA_4+\text{Ker}\phi$, $\cdots$, $x^{t_4'-t_4-1}A_4+\text{Ker}\phi$\}. To show $B$ spans $B_1$, we only show that $x^{t_1-t_2}A_2 \in \text{Span}(B)$. In a similar way, we can show that $x^{t_1-t_2+1}A_2$, $\cdots$, $x^{n-t_2-1}A_2,$ $\cdots$,  $x^{t_4'-t_4}A_4,\cdots, x^{n-t_4-1}A_4 \in \text{Span}(B).$ Since 
$B_{\phi}$ spans $C/\text{Ker}\phi$, we can write $x^{t_1-t_2}A_2+\text{Ker}\phi$ as  a  $R_{u^2,v^2,p}$ linear combination of the elements of $B_{\phi}$, i.e., $x^{t_1-t_2}A_2+\text{Ker}\phi$ = $\sum\limits_{i=0}^{n-t_1-1}\alpha_{i1}(x^iA_1+\text{Ker}\phi)$ + $\cdots$ + $\sum\limits_{i=0}^{t_4'-t_4-1}\alpha_{i4}(x^iA_4+\text{Ker}\phi$), where $c_{ij} \in R_{u^2,v^2,p}$. Thus, $x^{t_1-t_2}A_2$ - $(\sum\limits_{i=0}^{n-t_1-1}\alpha_{i1}(x^iA_1)$ + $\cdots$ + $\sum\limits_{i=0}^{t_4'-t_4-1}\alpha_{i4}(x^iA_4))$ $\in \text{Ker}\phi$. Since $\text{Ker}\phi = \text{Span}(B_2)$ and $B$ spans $B_2$, we get $x^{t_1-t_2}A_2 \in \text{Span}(B)$. Similarly, we can show that $x^{t_1-t_2+1}A_2$, $\cdots$, $x^{n-t_2-1}A_2,$ $\cdots$,  $x^{t_4'-t_4}A_4, \cdots, x^{n-t_4-1}A_4 \in \text{Span}(B).$ This shows that $B$ spans $B_1$.
It is easy to see that any elements of the spanning set $B$ can not be written as the linear combination of its preceding elements and other elements in the spanning set $B$. Here we only show that $x^{t_1-t_3-1}A_2$ can not be written as linear combinations of others element of spanning set $B$. The proof is similar for the rest. Suppose, if possible $x^{t_1-t_2-1}A_2$ can be written as linear combinations of the others element of the spanning set $B$. Then we have 
$x^{t_1-t_3-1}A_3=\sum_{i=0}^{n-t_1-1}\alpha_{1i}x^iA_1+\sum_{i=0}^{t_1-t_2-1}\alpha_{2i}x^iA_2+\sum_{i=0}^{t_1-t_3-2}\alpha_{3i}x^iA_3+\sum_{i=0}^{t_4'-t_4-1}\alpha_{4i}x^iA_4+\sum_{i=0}^{t_1-t_5-1}\alpha_{5i}x^iA_5+\sum_{i=0}^{t_6'-t_6-1}\alpha_{6i}x^iA_6+\sum_{i=0}^{t_7'-t_7-1}\alpha_{7i}x^iA_7+\linebreak\sum_{i=0}^{t_8'-t_8-1}\alpha_{8i}x^iA_8$, where, $\alpha_{ji}=\beta_{j1}^{(i)}+u\beta_{j2}^{(i)}+v\beta_{j3}^{(i)}+uv\beta_{j4}^{(i)}+w\beta_{j5}^{(i)}+uw\beta_{j6}^{(i)}+vw\beta_{j7}^{(i)}+uvw\beta_{j8}^{(i)} \in \F_p$ (Note that $i$ is not a power of $\beta$ it is a notation). We have $x^{t_1-t_3-1}(vf_3(x)+uvf_{3,4}(x)+wf_{3,5}(x)+uwf_{3,6}(x)+vwf_{3,7}(x)+uvwf_{3,8}(x))=f_1(x)\sum_{i=0}^{n-t_1-1}\beta_{11}^{(i)}x^i+uf_1(x)\sum_{i=0}^{n-t_1-1}\beta_{12}^{(i)}x^i+uf_{1,2}(x)\sum_{i=0}^{n-t_1-1}\beta_{11}^{(i)}x^i+uf_2(x)\sum_{i=0}^{t_1-t_2-2}\beta_{21}^{(i)}x^i+\linebreak vf_1(x)\sum_{i=0}^{n-t_1-1}\beta_{13}^{(i)}x^i+ vf_{1,3}(x)\sum_{i=0}^{n-t_1-1}\beta_{11}^{(i)}x^i+vf_{2,3}(x)\sum_{i=0}^{t_1-t_2-1}\beta_{21}^{(i)}x^i+\linebreak vf_3(x)\sum_{i=0}^{t_1-t_3-2}\beta_{31}^{(i)}x^i+uvm_2(x)+wm_3(x)+uw(m_4(x)+vwm_5(x)+uvwm_6(x)$, where, $m_2(x), \cdots, m_6(x)$ is a polynomials in $\F_p[x]$. By comparing both sides, we have $\beta_{11}^{(i)}=0$, $\beta_{12}^{(i)}=0$ for $0 \leq i \leq n-t_1-1$, $\beta_{21}^{(i)}=0$  for $0 \leq i \leq t_1-t_2-1$, and $x^{t_1-t_3-1}f_3(x)=f_{1}(x)\sum_{i=0}^{n-t_1-1}\beta_{13}^{(i)}x^i+f_3(x)\sum_{i=0}^{t_1-t_3-2}\beta_{31}^{(i)}x^i$. Note that $\text{deg}(x^{t_1-t_3-1}f_3(x))=t_1-1$ but $\text{deg}(f_1(x)\sum_{i=0}^{n-t_1-1}\beta_{13}^{(i)}x^i) \geq t_1$ and $\text{deg}(f_3(x)\sum_{i=0}^{t_1-t_3-2}\beta_{31}^{(i)}x^i) \leq t_1-2$. Hence, this gives a contradiction.
\end{proof}

\begin{theo}
 Let $n$ be a positive integer relatively prime to $p$ and $C$ be a cyclic code of length $n$ over the ring $R_{u^2,v^2,w^2,p}$. If $C=\langle f_1(x)+uf_2(x), vf_3(x)+uvf_4(x), w(f_5(x)+uf_6(x)), w(vf_7(x)+uvf_8(x))\rangle$ with $t_i={\rm deg}(f_i(x))$, $1 \leq i \leq 8$, and $ t_7'={\rm min}\{t_3,t_5\}, $ then $ C $ has rank $ n+t_1+t_7'-t_3-t_5-t_7$. The minimal spanning set $B$ of the code $ C $ is  $ B=\{f_1(x)+uf_2(x), x(f_1(x)+uf_2(x)), $ $\cdots, x^{n-t_1-1}(f_1(x)+uf_2(x)), vf_3(x)+uvf_4(x), x(vf_3(x)+uvf_4(x)),\cdots,$ $ x^{t_1-t_3-1}(vf_3(x)+uvf_4(x)),$  $w(f_5(x)+uf_6(x)), xw(f_5(x)+uf_6(x)), $ $\cdots, x^{t_1-t_5-1}w(f_5(x)+uf_6(x)), w(vf_7(x)+uvf_8(x)), xw(vf_7(x)+uvf_8(x)),\cdots,$ $x^{t_7'-t_7-1}w(vf_7(x)+uvf_8(x))\}$. 
\end{theo}
\begin{proof}
 The proof is similar to the above theorem.
\end{proof}

\section{Minimum distance} \label{md}
Let $n$ be a positive integer not relatively prime to $p$. Let $ C $ be a cyclic code of length $n$ over $R_{u^2,v^2,w^2,p}$. From Eq.\eqref{c8}, we have $C_{8}=\{f(x)\in \F_p[x]~|~uvwf(x)\in C\}=\langle f_{8}(x)\rangle$. Also, we know that $C_{8}$ is a cyclic code over $\F_p$.
\begin{theorem} \label{md1}
Let $n$ be a positive integer not relatively prime to $p$. If $ C = \langle A_1, A_2, \cdots, A_{8}\rangle$ is a cyclic code of length $n$ over $R_{u^2,v^2,w^2,p}$. Then $w_{H}(C) = w_{H}(C_{8})$.
\end{theorem}
\begin{proof}
Let $M(x)=m_0(x)+um_1(x)+vm_2(x)+uvm_3(x)+wm_4(x)+uwm_5(x)+vwm_6(x)+uvwm_7(x) \in C,$ where $m_0(x), m_1(x), \cdots, m_{7}(x) \in \F_p[x]$. We have $uvwM(x)=uvwm_0(x)$, $w_{H}(uvwM(x)) \leq w_{H}(M(x))$ and $uvwC$ is subcode of $C$ with $w_{H}(uvwC) \leq w_{H}(C)$. Thus  $w_{H}(uvwC)=w_{H}(C)$. Therefore, it is sufficient to focus on the subcode $uvwC$ in order to prove the theorem. Since $w_{H}(C_{8})=w_{H}(uvwC)$, we get $w_{H}(C)=w_{H}(C_{8})$.
\end{proof}
\begin{definition}
Let $ m = b_{l-1}p^{l-1} + b_{l-2}p^{l-2} + \cdots + b_1p + b_0$, $b_i \in \F_p, 0 
\leq i \leq l-1$, be the $p$-adic expansion of $m$.
\begin{enumerate} [{\rm (1)}]
 \item If $ b_{l-i}  \neq 0$ for all $1  \leq i \leq q, q < l, $ and $ b_{l-i} = 0 $ for all $i, q+1 \leq i \leq l$, then $m$ is said to have a $p$-adic length $q$ zero expansion.
\item If $ b_{l-i}  \neq 0$ for all $1  \leq i \leq q, q < l, $ $b_{l-q-1} = 0$ and $ b_{l-i} \neq 0 $ for some $i, q+2 \leq i \leq l$, then $m$ is said to have  $p$-adic length $q$ non-zero expansion.
\item If $ b_{l-i}  \neq 0$ for $1  \leq i \leq l, $ then $m$ is said to have a $p$-adic length $l$  expansion or $p$-adic full expansion.
\end{enumerate}
\end{definition}
\begin{lemma} \label{lm-md}
Let $C$ be a cyclic code over $R_{u^2,v^2,w^2,p}$ of length $p^l$ where $l$ is a positive integer. Let $C = \langle f(x)\rangle$ where $f(x) = (x^{p^{l-1}} - 1)^bh(x)$, $ 1 \leq b < p$. If $h(x)$ generates a cyclic code of length $p^{l-1}$ and minimum distance $d$ then the minimum distance $d(C)$ of $C$ is $(b+1)d$.
\end{lemma}
\begin{proof}
For $c \in C$, we have $c=(x^{p^{l-1}}-1)^bh(x)m(x)$ for some $ m(x) \in \frac{R_{u^2,v^2,w^2,p}[x]}{(x^{p^l}-1)}$. Since $h(x)$ generates a cyclic code of length $p^{l-1}$, we have $w(c) = w((x^{p^{l-1}} - 1)^bh(x)m(x)) = w(x^{p^{l-1}b}h(x)m(x)) + w(^bC_1x^{p^{l-1}(b-1)}h(x)m(x)) + \cdots + w(^bC_{b-1}x^{p^{l-1}}h(x)m(x)) + w(h(x)m(x))$. \\Thus, $ d(C) = (b + 1)d$.
\end{proof}

\begin{theorem} \label{md-thm}
Let $C$ be a cyclic code over $R_{u^2,v^2,w^2,p}$ of length $p^l$, where $l$ is a positive integer. Then,  $C = \langle A_1, A_2, \cdots, A_{8}\rangle$ where $f_1(x) = (x-1)^{t_1}, f_2(x) = (x-1)^{t_2}, \cdots,  f_{8}(x) = (x-1)^{t_{8}}$ for some $t_1>t_2,t_3>t_4>t_8>0$, $t_2 >t_6$, $t_3 >t_7$ and $t_1>t_5 >t_6,t_7> t_8>0$.
\begin{enumerate}[{\rm (1)}]
\item If $t_{8} \leq p^{l-1},$ then $d(C) = 2$. 
\item If $t_{8} > p^{l-1}$, let $t_{8} = b_{l-1}p^{l-1} + b_{l-2}p^{l-2} + \cdots + b_1p + b_0$ be the $p$-adic expansion of $t_{8}$ and $ f_{8}(x) = (x-1)^{t_{8}} = (x^{p^{l-1}} - 1)^{b_{l-1}}(x^{p^{l-2}} - 1)^{b_{l-2}} \cdots (x^{p^{1}} - 1)^{b_1}(x^{p^0} - 1)^{b_0}$.
\begin{enumerate}[{\rm ($a$)}]
 \item If $t_{8}$ has a $p$-adic length $q$ zero expansion or full expansion $(l=q)$, then $d(C) = (b_{l-1}+1)(b_{l-2}+1)\cdots(b_{l-q}+1)$.
\item If $t_{8}$ has a $p$-adic length $q$ non-zero expansion, then $d(C) = 2(b_{l-1}+1)(b_{l-2}+1)\cdots(b_{l-q}+1)$.
\end{enumerate}
\end{enumerate}
\end{theorem}
\begin{proof}
The first claim easily follows from Theorem \ref{properties}. From Theorem \ref{md1}, we see that $d(C)=d(C_{8})=d(\langle(x-1)^{t_{8}}\rangle)$. Hence, we only need to determine the minimum weight of $C_{8}= \langle(x-1)^{t_{8}}\rangle$.\\
(1) If $t_{8} \leq p^{l-1},$ then $(x-1)^{t_{8}}(x-1)^{p^{l-1}-t_{8}} = (x-1)^{p^{l-1}}=(x^{p^{l-1}}-1) \in C$. Thus, $d(C)=2$.\\
(2) Let $ t_{8}>p^{l-1}$. (a) If $t_{8}$ has a $p$-adic length $q$ zero expansion, we have $t_{8}=b_{l-1}p^{l-1}+b_{l-2}p^{l-2} + \cdots + b_{l-q}p^{l-q}$, and $f_{8}(x)=(x - 1)^{t_{8}}=(x^{p^{l-1}}-1)^{b_{l-1}}(x^{p^{l-2}}-1)^{b_{l-2}}\cdots(x^{p^{l-q}}-1)^{b_{l-q}}$. Let $h(x)=(x^{p^{l-q}}-1)^{b_{l-q}}$. Then $h(x)$ generates a cyclic code of length $p^{l-q+1}$ and minimum distance $(b_{l-q}+1)$. By Lemma \ref{lm-md}, the subcode generated by $(x^{p^{l-q+1}}-1)^{b_{l-q+1}}h(x)$ has minimum distance $(b_{l-q+1}+1)(b_{l-q}+1)$. By induction on $q$, we can see that the code generated by $f_{8}(x)$ has minimum distance $(b_{l-1}+1(b_{l-2}+1)\cdots(b_{l-q}+1)$. Thus, $d(C)=(b_{l-1}+1)(b_{l-2}+1)\cdots(b_{l-q}+1)$.\\
(b) If $t_{8}$ has a $p$-adic length $q$ non-zero expansion, we have $t_{8}=b_{l-1} p^{l-1} + b_{l-2}p^{l-2} + \cdots + b_{1}p + b_0, b_{l-q-1}=0$. Let $r=b_{l-q-2}p^{l-q-2}+b_{l-q-3}p^{l-q-3}+ \cdots + b_1p + b_0$ and $h(x)=(x-1)^r=(x^{p^{l-q-2}}-1)^{b_{l-q-2}}(x^{p^{l-q-3}}-1)^{b_{l-q-3}}\cdots(x^{p^{1}}-1)^{b_{1}}(x^{p^{0}}-1)^{b_{0}}$. Since $r < p^{l-q-1}$, we have $p^{l-q-1}=r+j$ for some non-zero $j$. Thus, $(x-1)^{p^{l-q-1}-j}h(x)=(x^{p^{l-q-1}}-1) \in C$. Hence, the subcode generated by $h(x)$ has minimum distance 2. By Lemma \ref{lm-md}, the subcode generated by $(x^{p^{l-q}}-1)^{b_{l-q}}h(x)$ has minimum distance $2(b_{l-q}+1)$. By induction on $q$, we can see that the code generated by $f_{8}(x)$ has minimum distance $2(b_{l-1}+1)(b_{l-2}+1)\cdots(b_{l-q}+1)$. Thus, $d(C) = 2(b_{l-1}+1)(b_{l-2}+1)\cdots(b_{l-q}+1)$.\\
\end{proof}

\section{Examples} \label{exm}
\begin{example}
Cyclic codes of length $4$ over the ring $R_{u^2,v^2,w^2,2}$. We have
\begin{equation}
x^4-1=(x-1)^4 ~ \text{over} ~ \F_2\notag
\end{equation}
Let $g=x-1$ . The some of the non zero cyclic codes of length 4 over the ring $R_{u^2,v^2,w^2,2}$ with generator polynomials, rank and minimum distance are given in Tables 1 and Binary images of some cyclic codes of length 4 over $R_{u^2,v^2,w^2,2}$ are given in Tables 2.

\end{example}

{\bf Table 1.} Non zero cyclic codes of length 4 over $R_{u^2,v^2,w^2,2}$.\\
\begin{center}
\begin{tabular}{| l | c | c |}
\hline
Non-zero generator polynomials & Rank & d(C)\\
\hline
$\langle vwg^2+(c_0+c_1x)uvw \rangle $& 2 & 2 \\
\hline
$\langle vwg+c_0uvw \rangle$& 3 & 2\\
\hline
$\langle uwg^3+c_1vwg^3+c_0uvwg^2 \rangle $& 1 & 4\\
\hline
$\langle uwg^3+c_0uvw(c_2+c_3x),vwg^3+(c_0+c_1x)uvw \rangle $ & 3 & 4\\
\hline
$\langle uwg^2+c_0vwg^2$, $uvw \rangle $& 4 & 1\\
\hline
$\langle wg^3+c_2uwg^2+c_3vwg^2+c_4uvw$, $uvwg^2 \rangle $& 2 & 2\\
\hline
$\langle wg^3+c_4uwg+c_5vwg+uvw(c_6+c_7x)$,$uwg^2+uvw$& 4 & 2\\
$(c_2+c_3x)$, $vwg^2+uvw(c_0+c_1x)$, $uvwg \rangle$& &\\
\hline
$ \langle wg+c_1uw+c_2vw \rangle$& 3 & 2\\
\hline
$ \langle uvg^3+c_1uwg^3+c_2vwg+c_3uvwg$,$vwg^2+c_0uvwg \rangle$& 3 & 2\\
\hline
$\langle vg^3+c_1uvg+c_1vwg$,$uvg^2+c_0wg,$ $wg^2$,$uw$,$vw \rangle$ & 8 & 1\\
\hline
$\langle ug^2+vg+c_1uv+c_0wg$,$vg^2+c_1uv+c_0wg$,$uvg+c_1wg$,$wg^2+c_2uw$, & 9 & 1\\
$uwg$, $vw \rangle$ &&\\
\hline
$\langle g^2+v+u+c_1w$,$vg+u$,$ug+c_1w$,$uv+c_1w$,$wg$,$uw$,$vw \rangle$& 8 & 1\\
\hline

\end{tabular}
\end{center}

\begin{center}
{\bf Table 2.} Binary images of some  cyclic codes of length 4 over $R_{u^2,v^2,w^2,2}$.
\begin{tabular}{| l | c| c |}
\hline
Non-zero generator polynomials &$\phi_L(C)$\\
\hline
$\langle uvwg^3\rangle$ & [32,1,32]*\\
\hline
$\langle vwg^3+uvwg^2 \rangle$ & [32,2,16]\\
\hline
$\langle uwg^3+vwg^2+uvwg,uvwg^2  \rangle$ &$[32,3,16]^{*-2}$\\
\hline
$\langle uvg^3+uwg^2+vwg^2+uvwg,uwg^3+vwg^3 \rangle$ &[32,4,16]*\\
\hline
$\langle uvg^3+uwg^2+vwg^2+uvw,uwg^3+vwg^3,uvwg \rangle$ &[32,5,16]*\\
\hline
\end{tabular}
\end{center}

\begin{example}
Cyclic codes of length $3$ over the ring $R_{u^2,v^2,w^2,3}$. We have
\begin{equation}
x^3-1=(x-1)^3 ~ \text{over} ~ \F_3\notag
\end{equation}
Let $g=x-1$ . The some of the non zero cyclic codes of length 3 over the ring $R_{u^2,v^2,w^2,3}$ with generator polynomials, rank and minimum distance are given in Tables below:

\end{example}

{\bf Table 3.} Non zero cyclic codes of length 3 over $R_{u^2,v^2,w^2,3}$.
\begin{center}
\begin{tabular}{| l | c| c |}
\hline
Non-zero generator polynomials & Rank & d(C)\\
\hline
$\langle uwg^2+vwg,vwg^2 \rangle$& 3 & 3 \\
\hline
$\langle uwg+c_2vw+c_1uvw,vwg+c_0uvw \rangle$& 4 & 2 \\
\hline
$\langle wg^2+uw+vw,uwg+vw,vwg,uvw \rangle$& 4 &1 \\
\hline
$\langle wg+c_1uw+c_2vw,uwg+c_0vw \rangle$& 3 & 2\\
\hline
$\langle uvg+c_2uw+c_1vw,wg+c_0uw \rangle$& 5 &1 \\
\hline
$\langle vg^2+c_4uvg+c_3w+c_2uw,uvg+wg^2+c_1uw,uwg^2+c_0vw,vwg,uvw \rangle$& 5 & 1\\
\hline
$\langle ug+c_3v+c_2wg+c_1vw,vg+c_1wg+c_1vw,uv+c_0wg$& 8 & 3 \\
$+c_1vw+c_0vw,wg^2+c_0vw+c_0uvw \rangle$& &\\
\hline
$\langle g^2+c_3u+c_0uvw,ug+c_2w+c_0uvw,v+c_1w+c_0uvw,wg+c_0uvw,$ & 6 &1 \\
$uwg+c_0uvw \rangle$ &&\\
\hline
\end{tabular}
\end{center}

\begin{center}
{\bf Table 4.} Ternary images of some  cyclic codes of length 3 over $R_{u^2,v^2,w^2,3}$.\\
\begin{tabular}{| l | c| c |}
\hline
Non-zero generator polynomials & $\phi_L(C)$ \\
\hline
$\langle uwg^2+2vwg^2+uvwg \rangle$  & [24,1,24]*  \\
\hline
$\langle uwg^2+2vwg^2+uvwg \rangle$  & $[24,2,16]^{*-2} $\\
\hline
$\langle uvg^2+uwg^2+vwg^2+2uvw,uvwg \rangle$  & $[24,3,15]^{*-1} $\\
\hline
$\langle uvg^2+wg^2+uwg+vwg+2uvwg,uwg^2+2uvwg,vwg^2+uvwg)\rangle$ & $[24,5,12]^{*-2} $ \\
\hline
$\langle uvg+wg^2+uwg+vwg,uwg^2+vwg,vwg^2)\rangle$ & $[24,7,8] $ \\
\hline
$\langle g^2+ug+vg+wg+uw+vw,uvg+wg+uw+vw,uwg+vw,vwg)\rangle$ & $[24,14,4]^{*-2} $ \\
\hline
\end{tabular}
\end{center}

\begin{center}
{\bf Table 5.} Ternary images of some Non zero free cyclic codes of length 3 over $R_{u^2,v^2,w^2,3}$.\\
\begin{tabular}{| l | c |}
\hline
Non-zero generator polynomials & $\phi_L(C)$\\
\hline
$\langle g^2+ug+uv+uw+vw\rangle$  &[24,8,7]\\
\hline
$\langle g+u+v+uv+w+uw+vw+uvw\rangle$ & [24,16,4]\\
\hline
\end{tabular}
\end{center}

\begin{example}
Cyclic codes of length $5$ over the ring $R_{u^2,v^2,w^2,5}$. We have
\begin{equation}
x^5-1=(x-1)^5 ~ \text{over} ~ \F_5\notag
\end{equation}
Let $g=x-1$ . The some of the non zero cyclic codes of length 5 over the ring $R_{u^2,v^2,w^2,5}$ with generator polynomials, rank and minimum distance are given in Tables below:

\end{example}

{\bf Table 6.} Non zero cyclic codes of length 5 over $R_{u^2,v^2,w^2,5}$.
\begin{center}
\begin{tabular}{| l | c| c |}
\hline
Non-zero generator polynomials & Rank & d(C)\\
\hline
$\langle vwg^4+uvwg^3\rangle$& 1 &5 \\
\hline
$\langle uwg^4+vwg^3+uvwg$,$vwg^4+uvwg^2$, $uvwg^3\rangle$& 3 &4 \\
\hline
$\langle uwg^3+c_4vwg+uvw(c_2+c_3x)$,$vwg^2+(c_0+c_1x)uvw \rangle$& 5 &3 \\
\hline
$\langle wg^4+uw(c_6+c_7x)+vwc_8+uvwc_9$,$uwg^2+vw(c_4+c_5x)+$& 5 &3 \\
$uvw(c_2+c_3x)$,$vwg^2+uvw(c_0+c_1x)\rangle$&  & \\
\hline
$\langle uvg^4+wc'_2g^3+uw(c'_0+c'_1x)+vw(c'_0+c'_1x)+uvw$,$wg^4+c_9uwg+$& 6 &3 \\
$vw(c_6+c_7x)+c_8uvw,uwg^2+vw(c_3+c_4x+c_5x^2)+uvw(c_1+c_2x)$, &  &\\
$vwg^2+uvw(c_0+c_1x)\rangle$ & &\\
\hline
$\langle uvg^3+c'_4wg^2+c'_3uw+c'_2vw$,$wg^4+uwg+c'_1vw+c'_0uvw$,& 6 & 3\\
$uwg^3+vw(c_3+c_4x)+c_5uvw$,$vwg^3+(c_0+c_1x)uvw$, $uvwg^2\rangle$ & &\\
\hline
$\langle vg^4+c'_0uvg+w(c_1+c'_2x)+c'_3uw+c'_4vw$,$uvg^2+c'_7wg^2+c'_8uw$& 8  &2 \\
$+c'_9vw+c'_5uvw$,$wg^4+c_4uwg^2+c_5vw+c_6uvw$,$uwg^3+vw$& &\\
$+uvwg ,vwg+(c_0+c_1x)uvw \rangle$& &\\
\hline

\end{tabular}
\end{center}

\begin{center}
{\bf Table 2.} 5-ary images of some  cyclic codes of length 5 over $R_{u^2,v^2,w^2,5}$.\\
\begin{tabular}{| l | c| c |}
\hline
Some non-zero generator polynomials & $\phi_L(C)$\\
\hline
$\langle uuvg^4\rangle$  & [40,1,40]* \\
\hline
$\langle vwg^4+uvwg^3 \rangle$  & $[40,2,32]^{*-1} $ \\
\hline
$\langle vwg^4+3uvwg^2,uvwg^3 \rangle$  & [40,3,28]\\
\hline
$\langle uvg^4+2uwg^3+vwg^4+uvwg^2,uwg^4+vwg^4+4uvwg^2 \rangle$  & $[40,4,28]^{*-2} $  \\
\hline
$\langle g,u+3v+w\rangle$ &  [40,37,2]* \\
\hline
$\langle g,uv+3w+2vw+uw\rangle$ &  $[40,36,2]^{*-2} $ \\
\hline

\end{tabular}
\end{center}

\begin{center}
{\bf Table 3.} 5-ary images of some Non zero free cyclic codes of length 5 over $R_{u^2,v^2,w^2,5}$.\\
\begin{tabular}{| l | c |}
\hline
Non-zero generator polynomials &$ \phi_L(C)$\\
\hline
$\langle g^4+2ug^3+v(3+2x+x^2)+uv(1+3x+4x^2)+w(3x+x^2)$ & [40,8,18]\\
$+uw(3+x+x^2)+vw(1+2x+x^2)+uvw(3x+x^2)\rangle$ & \\
\hline
$\langle g^3+uc_6g^2+vc_5(1+3x)+uvc_4(3+x)+wc_3(1+4x)+$ &  [40,16,12]\\
$uwc_2(2+x)+vwc_1(3+4x)+uvw(2+3x)\rangle$ & \\
\hline
$\langle g^2+u(1+4x)+3uv+uw+vw\rangle$  &  $[40,24,8]^{b-1} $\\
\hline
$\langle g+4u\rangle$  &  [40,31,3]\\
\hline
\end{tabular}
\end{center}

\bibliographystyle{plain}
\bibliography{ref}
\end{document}